\documentclass[a4paper,final]{lipics-v2021}

\usepackage{xspace}
\usepackage{amssymb,amsmath,amsfonts,amsthm}
\usepackage{ifthen}
\usepackage{graphicx}
\usepackage{MnSymbol}
\usepackage[dvipsnames]{xcolor}
\usepackage[color]{showkeys}

\usepackage{tikz}

\usetikzlibrary{decorations}
\usetikzlibrary{shapes}
\usetikzlibrary{arrows,automata,positioning}

\title{\texorpdfstring{Quadratic word equations with regular constraints\\ and the exponent of periodicity}{Quadratic word equations with regular constraints and the exponent of periodicity}}
\titlerunning{Quadratic word equations and the exponent of periodicity}

\author{Volker Diekert}{University of Stuttgart, Germany}{diekert@fmi.uni-stuttgart.de}{https://orcid.org/0000-0002-5994-3762}{}
\author{Silas Natterer}{University of Stuttgart, Germany}{natterersilas@gmail.com}{https://orcid.org/0009-0004-6524-0854}{}
\author{Alexander Thumm}{University of Siegen, Germany}{alexander.thumm@uni-siegen.de}{https://orcid.org/0009-0005-4240-2045}{}

\authorrunning{V.~Diekert, S.~Natterer, A.~Thumm}
\Copyright{V.~Diekert, S.~Natterer, A.~Thumm}


\acknowledgements{The authors thank the anonymous reviewer for valuable comments and suggestions that helped improve the presentation of this paper. We also thank Jorge Almeida, Olga Kharlampovich, Alexei Miasnikov, Wojciech Plandowski, Zlil Sela, and the GAGTA community for various discussions and their interest in the subject.}

\nolinenumbers
\hideLIPIcs

\usepackage[utf8]{inputenc}
\usepackage{verbatim,wrapfig}

\usepackage{bold-extra}

\usepackage{amsmath,amssymb,bm}
\usepackage{mathtools}

\usepackage{enumerate}
\usepackage{multirow}
\usepackage{xspace}
\usepackage{tikz}
\usetikzlibrary{decorations.pathreplacing}
\usetikzlibrary{automata, positioning, arrows}

\usepackage{tcolorbox}
\usepackage{dsfont}
\usepackage{textcomp}

\usepackage{index}

\makeindex

\usepackage[ruled]{algorithm}
\usepackage[noend]{algpseudocode}

\newtheorem{problem}{Problem}

\newcommand{\Sol}{\mathrm{Sol}}

\newcommand{\stab}{\operatorname{stab}}

\renewcommand{\implies}{\Rightarrow}
\renewcommand{\ast}{*}
\renewcommand{\colon}{:}

\renewcommand{\iff}{\mathrel{\Leftrightarrow}}

\newcommand{\Sop}{S^{\mathrm{op}}}

\newcommand{\pv}[1]{\ensuremath{\mathbf{#1}}\xspace}

\newcommand{\lds}{, \ldots ,}

\newcommand{\Mati}{Matiyasevich\xspace}
\newcommand{\MatiSG}{Matiyasevich solution graph\xspace}

\newcommand{\nondet}{non\-deter\-ministic\xspace}

\newcommand{\IFF}{if and only if\xspace}
\renewcommand{\hom}{homo\-mor\-phism\xspace}
\newcommand{\homs}{homo\-mor\-phisms\xspace}
\newcommand{\Endo}{endo\-mor\-phism\xspace}
\newcommand{\Endos}{endo\-mor\-phisms\xspace}

\newcommand{\morph}{mor\-phism\xspace}

\DeclareMathOperator{\End}{\mathrm{End}}

\newcommand{\cXsing}{\cX_{\mathrm{sing}}}

\newcommand{\tra}{transition\xspace} 
\newcommand{\tras}{transitions\xspace}
\newcommand{\epstra}{$\eps$-transition\xspace}

\newcommand{\eg}{e.g.\xspace}

\newcommand{\Ip}{In parti\-cu\-lar,\xspace}
\newcommand{\ip}{in parti\-cu\-lar,\xspace}
\newcommand{\invol}{involution\xspace}

\newcommand{\solu}{solu\-tion\xspace}
\newcommand{\solus}{solu\-tions\xspace}

\newcommand{\subst}{sub\-sti\-tu\-tion\xspace}
\newcommand{\substs}{sub\-sti\-tu\-tions\xspace}

\usepackage{prettyref}
\newcommand{\prref}[1]{\prettyref{#1}}
\newrefformat{thm}{Theorem~\ref{#1}}
\newrefformat{lem}{Lemma~\ref{#1}}
\newrefformat{def}{Definition~\ref{#1}}
\newrefformat{claim}{Claim~\ref{#1}}
\newrefformat{conj}{Conjecture~\ref{#1}}
\newrefformat{conv}{Convention~\ref{#1}}
\newrefformat{cor}{Corollary~\ref{#1}}
\newrefformat{inva}{Invariant~\ref{#1}}
\newrefformat{prop}{Proposition~\ref{#1}}
\newrefformat{sec}{Section~\ref{#1}}
\newrefformat{kap}{Chapter~\ref{#1}}
\newrefformat{ex}{Example~\ref{#1}}
\newrefformat{eq}{Equation~(\ref{#1})}
\newrefformat{rem}{Remark~\ref{#1}}
\newrefformat{fig}{Figure~\ref{#1}}
\newrefformat{obs}{Observation~\ref{#1}}
\newrefformat{prob}{Problem~\ref{#1}}
\newrefformat{pro}{Property~\ref{#1}}
\newrefformat{ruli}{Rule~\ref{#1}}
\newrefformat{ruls}{Rules~\ref{#1}}

\newcommand{\wt}[1]{\widetilde{#1}}

\newcommand{\larc}[1]{\overset{#1\hspace*{2pt}}\longrightarrow}

\newcommand{\set}[2]
{\{#1\mid #2\}}

\newcommand{\os}[1]{\{#1\}}
\newcommand{\sm}{\setminus}
\newcommand{\es}{\emptyset}
\newcommand{\sse}{\subseteq}

\newcommand{\abs}[1]{\lvert#1\rvert}

\newcommand{\N}{\ensuremath{\mathbb{N}}}

\newcommand{\PSPACE}{\ensuremath{\mathsf{PSPACE}}}

\renewcommand{\P}{\ensuremath{\mathsf{PTIME}}}

\renewcommand{\phi}{\varphi}
\newcommand{\eps}{\varepsilon}

\newcommand{\oo}{\omega}
\newcommand{\alp}{\alpha}
\newcommand{\bet}{\beta}
\newcommand{\gam}{\gamma}

\newcommand{\sig}{\sigma}
\newcommand{\Sig}{\Sigma}
\newcommand{\Gam}{\Gamma}

\newcommand\OO{\Omega}

\newcommand{\Oh}{\mathcal{O}}

\newcommand{\cA}{\mathcal{A}}

\newcommand{\cC}{\mathcal{C}}
\newcommand{\cD}{\mathcal{D}}
\newcommand{\cE}{\mathcal{E}}
\newcommand{\cQE}{\mathcal{QE}}

\newcommand{\cH}{\mathcal{H}}

\newcommand{\cJ}{\mathcal{J}}
\newcommand{\cK}{\mathcal{K}}
\newcommand{\cL}{\mathcal{L}}

\newcommand{\cR}{\mathcal{R}}
\newcommand{\cS}{\mathcal{S}}

\newcommand{\cX}{\mathcal{X}}

\newcommand{\oi}[1]{{#1}^{-1}}

\newcommand{\Rat}{\mathrm{Rat}}
\newcommand{\Rec}{\mathrm{Rec}}
\newcommand{\Reg}{\mathrm{Reg}}

\newcommand{\siL}{\sim_\cL}
\newcommand{\siR}{\sim_\cR}

\newcommand{\siJ}{\sim_\cJ}

\newcommand{\leL}{\leq_\cL}
\newcommand{\leR}{\leq_\cR}
\newcommand{\leJ}{\leq_\cJ}

\newcommand{\geR}{\geq_\cR}
\newcommand{\geJ}{\geq_\cJ}

\begin{document}
\maketitle
\begin{abstract}
In this article, we study word equations in free semigroups and the conjecture that the existence of infinitely many solutions entails the existence of solutions with arbitrarily large exponent of periodicity.
We examine this question in the broader framework of word equations with regular constraints and establish new positive results: the conjecture holds for all quadratic word equations with constraints in finite semigroups from the variety $\mathbf{DLG}$ and its left-right dual $\mathbf{DRG}$, encompassing, in particular, all finite groups, commutative semigroups, and $\mathcal{J}$-trivial semigroups.
\end{abstract}
\begin{center}
\emph{This paper is dedicated to Mikhail V.~Volkov and his ``duo'' Mischa for being forever young.
%his seventieth birthday.
}
\end{center}
\section{Introduction}\label{sec:intro}
Word equations form a rich area of study.
A natural question in this general setting asks whether the existence of infinitely many \solu{s} implies the existence of \solu{s} with arbitrarily large exponent of periodicity. 
In this article, we examine a variant of this question for quadratic word equations with regular constraints, uncovering interesting connections between word equations and structural properties of finite semigroups along the way.

Word equations rose to prominence when Makanin proved that the existential theory of free semigroups is decidable. 
His algorithm drew attention on several counts: it resolved a long-standing open problem, the result was unexpected, and the accompanying termination proof stretched beyond 80~pages.
That proof rested on a theorem of Hmelevski\u{\i}, which guarantees that the shortest solution of a word equation has an exponent of periodicity below some computable bound.
This is the backdrop against which the story of our paper unfolds.

\subsection{A Conjecture}\label{sec:bcon}
Consider a word equation $U = V$ over constants and variables, and let $\sig$ be a \solu, that is, a substitution mapping variables to words over constants in such a way that $\sig(U) = \sig(V)$.

The \emph{exponent of periodicity} of~$\sig$, denoted by $\exp(\sig)$, is defined as the largest power of a nonempty word that appears as a factor in some $\sig(X)$, where $X$ is a variable.
(The formal definition is given in \prref{def:defexpinf}.)
This notion naturally leads to the following question.
\begin{center}
  \emph{Provided that a word equation admits infinitely many \solus,\\ must there exist \solus with arbitrarily large exponent of periodicity?}
\end{center}
The bold conjecture is that the answer to the above question is positive.\footnote{Unfortunately, bold conjectures share the fate of pilots: there are many old pilots and many bold pilots, but not so many who are both old and bold.}
The authors support this conjecture, for which several partial results provide evidence. 

First, the conjecture holds for \emph{quadratic equations}, meaning equations in which each variable occurs at most twice.
This is the content of \prref{prop:trivial}, which is subsumed by our main result, \prref{thm:main}, which also allows for certain types of regular constraints.
Furthermore, the conjecture also holds for the case of \emph{two-variable equations}, as can be derived from the description of their solution sets by Lucien Ilie and Wojciech Plandowski~\cite[Thm.~36]{IliePlandoswskiRAIRO2000}; see \prref{prop:IPR2000} for a restatement in our framework.

Additional support arises from group theory. 
In their work on Tarski’s conjectures on the elementary theory of free groups, Olga Kharlampovich and Alexei Miasnikov~\cite{KMIV06}, as well as Zlil Sela~\cite{sela13}, employ induction arguments based on the JSJ-decomposition of limit~groups.
According to these authors (personal communication), the same techniques show that the conjecture holds for equations in free groups, although this result has not been published yet.

Finally, a weaker form of support comes from computational exploration: so far we have not been able to construct a counterexample, neither manually nor with computer assistance.

\subsection{Historical Background}\label{sec:history}

The \emph{satisfiability problem for word equations} asks whether a given equation $U = V$ admits a \solu. 
Interest in this problem first arose in the Russian school of mathematics, when it was observed that this problem can be reduced to \emph{Hilbert's Tenth Problem}. 
Thus, proving the undecidability of the former would have simultaneously resolved the latter in the negative.

This program ultimately failed greatly with two major milestones in mathematics. 
In 1970, Yuri \Mati demonstrated the undecidability of Hilbert's Tenth Problem using number-theoretic methods building on the earlier work of Martin Davis, Hilary Putnam, and Julia Robinson; see~\cite{mat93} for a detailed exposition. 
The resulting theorem is now referred to as the DPRM theorem, in accordance with \Mati{}’s preference.

The second milestone was achieved in 1977 by Gennadi{\u\i} Makanin~\cite{mak77}, who proved that the satisfiability problem for word equations is, in fact, decidable. 
A key ingredient in his proof is a result by Hmelevski{\u\i}~\cite{hme76}, showing that a shortest \solu\ of a word equation~$U = V$ has an exponent of periodicity which is bounded in the size of the equation. 
Moreover, there is a computable number~$\kappa$, depending only on $|UV|$, such that if a \solu~$\sigma$ satisfies $\exp(\sigma) \ge \kappa$, then there are \solu{s}~$\sigma_n$ with $\exp(\sigma_n) \ge n$ for all $n \in \mathbb{N}$; we then write $\exp(U = V) = \infty$. 

Clearly, $\exp(U = V) = \infty$ implies the existence of infinitely many \solu{s}. 
But what about the converse direction?
This question behaves, to quote a phrase, like `\emph{a complete unknown, like a rolling stone}': it is easy to state and easy to grasp, and~--~with numerous new tools for analyzing word equations
having emerged since Makanin's result in 1977~--~it might now be within reach to answer this question.
However, it still tends to `\emph{roll away}'.\footnote{Or, in Salinger's terms, it resembles \emph{The Catcher in the Rye}.}

What are these new tools? 
First and foremost, \emph{compression techniques} have greatly clarified why the satisfiability problem for word equations is decidable. 
This approach was initiated by Wojciech Plandowski and Wojciech Rytter~\cite{pr98icalp} and simplified by Artur~Je\.{z}~\cite{Jez16jacm_stacs}. 
Moreover, compression techniques extend naturally to \emph{word equations with regular constraints}, as studied by Klaus Schulz in his Habilitationsschrift~\cite{schulz90}. 

Broadening the setting to include regular constraints has made the field more attractive to research that builds upon word equations. 
For instance, many modern works in \emph{string solving} and \emph{formal verification} incorporate word equations as part of their underlying framework. 

Solving equations with regular constraints also opens the possibility of finding an equation that has infinitely many constrained \solu{s}, yet~--~solely because of the constraints~--~the exponent of periodicity remains bounded by some fixed constant. 

For example, the word equation $XabY = YbaX$ (over constants $a,b$ and variables $X,Y$) is a promising candidate. 
Lucien Ilie and Wojciech Plandowski~\cite{IliePlandoswskiRAIRO2000} devoted an entire section to describing the full solution set of this short, two-variable quadratic equation. 
See also the article of Weinbaum~\cite{WeinbaumPacific2004ABCD} for the description of the \solu set of equations $XAY=YBX$ with $A\neq B$; and \prref{prob:4} for more details. 

On the other hand, the hope for positive results was kindled during a discussion with Wojciech Plandowski at a meeting held in Wittenberg, Germany, in September~2018. 
In that discussion, Plandowski referred to the equation $XabY = YbaX$ as the \emph{`one and only'} known equation with infinitely many \solu{s} that conceals its property $\exp(XabY = YbaX) = \infty$. 
His joke alluded to the fact that typical examples of word equations with infinitely many \solu{s} arise from systems of integer linear equations. 
This \emph{`one and only'} equation is, of course, the source of infinitely many similar examples.\footnote{Similarly, the language $\set{a^n b^n}{n \in \N}$ is the \emph{`one and only'} context-free language which is not regular.}

\subsection{Our Contribution}

In this article, we address the question posed above: 
whether a word equation with infinitely many solutions necessarily admits solutions of arbitrarily large exponent of periodicity. 
We study this problem in the setting of \emph{word equations in free semigroups with regular constraints}.

To obtain positive results, we approach the problem by considering restricted classes of word equations and restricted classes of finite semigroups. 
Given a class of word equations~$\cE$, we call a finite semigroup~$S$ \emph{nice} for~$\mathcal{E}$ if every equation in~$\cE$ with constraints in~$S$ has infinitely many solutions if and only if it has solutions with arbitrarily large exponent of periodicity; see~\prref{def:nicecon}.  
We show that, for every fixed class~$\cE$, the corresponding class of all finite semigroups that are nice for~$\cE$ enjoys several natural closure properties. 
Indeed, this class is always closed under taking subsemigroups and homomorphic images. 
Moreover, under mild assumptions on~$\cE$, it is also closed under taking left-right duals and forming retractive nilpotent extensions; see~\prref{lem:niceS}, \prref{lem:niceH}, and \prref{thm:ell} in \prref{sec:NSG}.

As an immediate application of these closure properties, we obtain that all finite nilpotent semigroups are nice for the class of \emph{two-variable word equations}; see \prref{cor:IPR2000}.

Our main result concerns the class of all \emph{quadratic word equations}~--~that is, equations in which each variable occurs at most twice~--~and members of the variety~$\pv{DLG}$ of finite semigroups whose regular 
$\mathcal{D}$-classes are right groups or, equivalently, whose regular $\mathcal{L}$-classes are groups (see \prref{sec:ddlg}).
The following statement, the proof of which occupies \prref{sec:sobeauty}, is formulated in~\prref{thm:mainiac}.

\begin{theorem}\label{thm:main}
  Let $S \in \pv{DLG}$.
  Then $S$ is nice for all quadratic word equations.
\end{theorem}

Our results thus support the conjecture in a nontrivial setting, where additional regular constraints are permitted. 
Moreover, the authors believe that this is a nice result.

\section{Notation and Preliminaries}\label{sec:nota}

By $\N$ we denote the set of natural numbers and by $\N_+$, the set of positive natural numbers.
Frequently, we identify an element with the singleton containing this element. 
The restriction of a mapping $\psi: B\to C$ to some $A \sse B$ is still denoted $\psi$.
The set of mappings from $A$ to $B$ is denoted by $B^A$.
If $\phi: A\to B$ and $\psi: B\to C$ are mappings then the composition $\psi \phi:A\to C$ is defined by $(\psi \phi)(a)= \psi(\phi(a))$. 

\subsection{Semigroups and Monoids}\label{sec:free}
A monoid is a semigroup with a neutral element, which we usually denote by~$1$. 
Every \hom between monoids has to map the neutral element to the neutral element. 
This is not required for semigroups. 
Thus, for monoids $S$ and $T$ there can be more semigroup \homs from $S$ to $T$ than monoid \homs. 

The \emph{empty} semigroup is a semigroup with no elements.
A semigroup with exactly one element is called \emph{trivial}.
For a given semigroup $S$, we let $S^1=S$ provided that $S$ is a monoid.
If $S$ is not a monoid, then $S^1$ is obtained by adjoining a neutral element to $S$. 
Thus, $S^1=S \cup \os 1$.
For every semigroup $S = (S,\cdot)$ there is a left-right dual semigroup $\Sop=(S,\circ)$, using the same underlying set $S$, but with its multiplication reversed, that is, $x \circ y= y \cdot x$. 
Note that, in contrast to groups, the semigroups $S$ and $\Sop$ are not isomorphic, in general. 

\smallskip

A \emph{pseudovariety} is a class of  semigroups which is closed under forming subsemigroups, homomorphic images, and finite direct products.\footnote{Consequently, every variety of finite semigroups contains the trivial semigroup and the empty one.}
Following Pin~\cite{pin86}, we call a pseudovariety of finite semigroups simply a \emph{variety of finite semigroups}. 
This facilitates the notation and there is no risk of confusion because if a variety contains a semigroup with at least two elements, it contains an infinite semigroup.

\smallskip

The \emph{free monoid} over a set $\Gam$ is denoted by $\Gam^*$.
In this context $\Gam$ is called an \emph{alphabet}, its elements are called \emph{letters}, and the elements of $\Gam^*$ are called \emph{words}.
The empty word is the neutral element and denoted by $1$. 
The \emph{free semigroup} over a set $\Gam$ is $\Gam^+ = \Gam^* \sm \os 1$.

If $w \in \Gam^*$ is a word, then $\abs{w}$ denotes its length and $\abs{w}_a$ denotes the number of occurrences of the letter $a \in \Gam$ within $w$ so that, in particular, $\abs{w} = \sum_{a \in \Gam} \abs{w}_a$. 

A word $u \in \Gam^*$ is a \emph{factor} of $w \in \Gam^*$ if there are $p,q \in \Gam^*$ such that $w = puq$.
Moreover, if $p$ or $q$ can be chosen as the empty word, then $u$ is a \emph{prefix} or \emph{suffix} of $w$, respectively.
These notions generalize to arbitrary semigroups in the form of Green's relations.

\subsection{Green's Relations}\label{sec:green}
We assume that most of the readers are familiar with Green's relations. 
For convenience we still recall some basic facts using the same notation as, for example, in \cite[Sec.~7.6]{edam16}.

Let $S$ be a semigroup.
For each $x \in S$ the left ideal $S^1x$, the right ideal $xS^1$, and the two-sided ideal $S^1xS^1$ are subsemigroups of $S$ containing~$x$. 
The concept of ideals leads to natural preorders $\leL$, $\leR$, $\leJ$ on $S$ and their induced equivalence relations $\siL$, $\siR$, $\siJ$: 
\begin{alignat*}{6}
  x &\leL y & \;\iff\;  &\phantom{S^1}S^1x &\,\sse\, & S^1y & \qquad
  x &\siL y & \;\iff\;  &\phantom{S^1}S^1x &\,=\,& S^1y\\
  x &\leR y & \;\iff\;  &\phantom{S^1}xS^1 &\,\sse\, & yS^1 & \qquad
  x &\siR y & \;\iff\;  &\phantom{S^1}xS^1 &\,=\,& yS^1\\
  x &\leJ y & \;\iff\;  &S^1xS^1 &\,\sse\, & S^1yS^1&\qquad
  x &\siJ y & \;\iff\;  &S^1xS^1 &\,=\,& S^1yS^1
\end{alignat*}

These are three out of five relations referred to as \emph{Green's relations}. 
Instead of $\siL$, $\siR$, and $\siJ$ we also write $\cL$, $\cR$, and $\cJ$ for the corresponding relations on~$S$. 
The other two relations are then defined by $\cH = \cL \cap \cR$ and $\cD = \cL \circ \cR =\cR \circ \cL$.
For any $\cK\in \os{\cH, \cL, \cR, \cD,\cJ}$ we denote the $\cK$-class of an element $x \in S$ by 
$\cK(x) = \set{y\in S}{x \sim_\cK y}$.

An important fact about finite semigroups is that they satisfy $\cD = \cJ$; thus, all~$\cD$-classes are $\cJ$-classes and vice versa. 
Moreover, if $S$ is finite, then there is some $\oo \in \N_+$ such that $x^\oo$ is idempotent for all $x\in S$, where an element $e \in S$ is called \emph{idempotent} if $e^2 = e$.

An element $x$ in a semigroup $S$ is called \emph{regular} if there exists $y \in S$ with $xyx = x$, and an $\cH$-,  $\cL$-, $\cR$-, or $\cD$-class is called \emph{regular} if it contains an idempotent element. 
The following is a well-known fact;  for a standard proof (for finite semigroups) see \eg \cite[Prop.~7.49]{edam16}.
\begin{proposition}\label{prop:Prop749edam}
Let $L$, $R$, and $D$ be $\cL$-, $\cR$-, and $\cD$-classes of a semigroup~$S$, respectively.
Further, suppose that $L, R \sse D$.
Then the following assertions are equivalent. 
\begin{bracketenumerate}
  \item The $\cD$-class $D$ is regular.
  \item The $\cL$-class $L$ is regular.
  \item The $\cR$-class $R$ is regular.
  \item Every element in $D$ is regular.
\end{bracketenumerate}
\end{proposition}

\subsection{Rational, Recognizable, and Regular Sets}\label{sec:regrat}
We now turn to rational, recognizable, and regular sets.
All facts discussed about them in this subsection admit simple proofs, which can be found, for example, in Eilenberg's classical textbook~\cite{eil74} or in other textbooks, sometimes stated for monoids, as in \cite[Sec.~7]{edam16}.

Let $S$ be an arbitrary semigroup.
Its family of \emph{rational sets} is the family $\Rat(S)$ of subsets of $S$ that is defined inductively as follows. 
\begin{itemize}
\item If $L \sse S$ is finite, then $L$ belongs to $\Rat(S)$.
\item If $L_1, L_2 \in \Rat(S)$, then $L_1\cup L_2$ and $L_1 \cdot L_2 = \set{u_1 \cdot u_2}{u_1 \in L_1, u_2 \in L_2}$ are in $\Rat(S)$.
\item If $L \in \Rat(S)$, then the subsemigroup $L^+ = L \cup L \cdot L \cup \cdots$
 generated by $L$ is in $\Rat(S)$.
\end{itemize}
There is a convenient alternative definition of $\Rat(S)$ by using \nondet automata.

\smallskip

A \emph{\nondet~$S$-automaton}, or simply an \emph{$S$-automaton} for short, is a directed graph $\cA$ with sets of initial and final vertices, and where arcs are labeled by elements of~$S$. 
\begin{itemize}
\item The vertices of $\cA$ are called \emph{states}, and the set of all states of $\cA$ is typically denoted by~$Q$.
There are distinguished sets $I\sse Q$ of \emph{initial states}, and $F\sse Q$ of \emph{final states}.
\item The directed labeled arcs of $\cA$ are called \emph{\tras}.
  Each \tra{} can be written as a triple $(p,s,q) \in Q \times S \times Q$ and depicted as $p \larc{s} q$. 
  Given any path $p_0 \larc{s_1} p_1 \; \cdots\; \larc{s_n} p_n$ of \tras{} in $\cA$, we define its label to be the element~$s_1 \cdots s_n\in S^1$. 
Note that this label is an element of $S$ unless $n = 0$, that is, unless the path is an \emph{empty path}.
\item The \emph{accepted language} $L(\cA)$ of $\cA$ is defined as follows.
\begin{align*}
L(\cA) = \set{s\in S}{\text{there is a path labeled } s \text{ from an initial to a final state in $\cA$}}
\end{align*}
\end{itemize}

An $S$-automaton $\cA$ with finitely many \tras{} is called a \emph{\nondet finite $S$-automaton}, or  
\emph{$S$-NFA} for short.
The family  $\Rat(S)$ coincides with the family of all languages which are accepted by $S$-NFA.

It is sometimes convenient to only consider those states of an $S$-automaton that directly contribute to its accepted language.
An~$S$-automaton is called \emph{trim} if every state is on some directed path from an initial to a final state.

\smallskip

For a semigroup~$S$ there is also the family $\Rec(S)$ of \emph{recognizable sets} of~$S$.
It consists of those subsets $L \sse S$ that are \emph{recognized} by some \hom $\mu : S \to T$ to a finite semigroup $T$, which means that $L = \oi \mu(\mu(L))$.
\begin{itemize}
\item A semigroup~$S$ is finitely generated \IFF $\Rec(S) \sse \Rat(S)$.
\item If $\phi: S \to S'$ is a \hom, then $L' \in \Rec(S')$ implies $\oi\phi(L')\in \Rec(S)$.
\item The family $\Rec(S)$ is a Boolean algebra with respect to the standard set operations. 
\item If $S$ contains the free product $\N_+ \ast (\N_+ \times \N_+)$, then $\Rat(S)$ is not closed under finite intersection.\footnote{In fact, closure of $\Rat(S)$ under finite intersection should be viewed as a rare exception.}
\Ip this implies $\Rat(S) \neq \Rec(S)$.
\end{itemize}

The descriptions of full \solu sets for word equations (in our and several other papers) use rational sets in the monoid of \Endos $\End(\Gam^+)$.
For $\abs{\Gam} \geq 2$ this monoid is not finitely generated, it is not torsion free, and it contains $\N_+ \ast (\N_+ \times \N_+)$ as a subsemigroup.\footnote{%
A possible realization of $\N_+ \ast (\N_+ \times \N_+)$ uses three \Endos $\alp$, $\bet$ and $\gam$. 
Here, $\gam$ is defined by $a\mapsto ab$ and $b\mapsto ba$, $\alp$ is defined by $a\mapsto a^2$ and $b\mapsto b$, and $\bet$ interchanges the role of $a$ and $b$ in $\alp$.}

Kleene's subset construction \cite{Kle56} shows that for every finitely generated free semigroup $S = \Gam^+$ (resp.\ for
 $S = \Gam^\ast$) it holds that $\Rec(S) = \Rat(S)$. 
In this case, the sets in $\Rec(S) = \Rat(S)$ are called \emph{regular sets}, and we also write $\Reg(S)$ for the family of all such sets.
We use this terminology only when $\Gam$ is finite.
The regular sets of finitely generated free semigroups and monoids are Boolean algebras and related by 
\[
  \Reg(\Gam^+)= \set{L \in \Reg(\Gam^*)}{1 \notin L},\;
  \Reg(\Gam^*) = \Reg(\Gam^+) \cup \set{ L \cup \os 1 \sse \Gam^*}{L \in \Reg(\Gam^+)}.
\]

\section{Word Equations}
In this paper, we deal with systems of word equations with regular constraints over finitely generated free monoids and semigroups.
We begin with the monoid case. 

\subparagraph*{Word Equations over Free Monoids.}
Let $\Sig$ and $\cX$ be two disjoint finite sets.
We say that $\Sig$ is the \emph{alphabet of constants} and $\cX$ is the set of \emph{variables}; by $\OO = \Sig\cup \cX$ we mean their (disjoint) union. 
A \emph{system of word equations (with regular constraints)} over $\OO^* = (\Sig \cup \cX)^*$ is a pair $(\cS,\mu)$ where $\cS = \os{(U_1,V_1) \lds (U_n,V_n)}$ is a finite set of \emph{word equations} $(U_i,V_i)$~--~that is, $U_i, V_i\in \OO^*$~--~and $\mu: \OO^* \to N$ is a \morph to a finite monoid~$N$.
 
A \emph{\solu} of $E=(\cS,\mu)$ is given by a mapping~$\sig: \cX \to \Sig^*$ such that its extension to a \morph~$\sig: \OO^* \to \Sig^*$, leaving the constants invariant, satisfies the following conditions.
\begin{bracketenumerate}
\item We have $\sig(U)=\sig(V)$ for all $(U,V)\in \cS$.
\item We have $\mu\sig(X)=\mu(X)$ for all $X\in \cX$; that is, \solus have to respect the constraints. 
\end{bracketenumerate}
The set of all \solu{s} of $E = (\cS, \mu)$ is denoted by $\Sol(E) = \Sol(\cS, \mu)$. 

\begin{remark}
In the literature, regular constraints are sometimes specified by a family of regular languages $L_X \sse \Sigma^\ast$ ($X \in \cX$), together with the requirement that every solution $\sigma$ satisfy $\sigma(X) \in L_X$ for all $X \in \cX$.
We now explain how this notion fits into our framework.

Since each language $L_X$ is regular, it is recognized by a homomorphism $\mu_X \colon \Sigma^\ast \to N_X$ to a finite monoid $N_X$.
Let $N = \prod_{X \in \cX} N_X$.
Then the homomorphism $\mu \colon \Sigma^\ast \to N$ whose components are the maps $\mu_X$ recognizes each of the languages $L_X$.
To fit our framework, we then extend $\mu$ to some homomorphism $\mu \colon \Omega^\ast \to N$ by choosing an image $\mu(X) \in \mu(L_X) \sse N$ for every $X \in \cX$.
Typically, regular constraints given by a family of regular languages thus correspond in our setting not to a single homomorphism, but rather to finitely many.
Indeed, for each variable $X$, the value $\mu(X)$ may be chosen arbitrarily from the set $\mu(L_X) \sse N$ and, since $\cX$ and $N$ are both finite, only finitely many such choices are possible.

The set of solutions of a system of equations $\cS$ with constraints given by the regular languages $L_X$ is then the union of the sets $\Sol(\cS,\mu)$ over all homomorphisms $\mu$ as above.
As this union is finite, it suffices for our purposes to consider a single homomorphism $\mu$.
\hspace*{\fill}$\diamond$\end{remark}

A \solu~$\sig$ is called \emph{singular} if $\sig(X)=1$ for some variable $X\in \cX$.
Otherwise~$\sig$ is called \emph{nonsingular}.
For each $\sig\in \Sol(E)$ there is a subset $\cXsing$ such that $\sig(X)=1\iff X\in \cXsing$.
By removing all variables in $\cXsing$ from the system~$E$ (that is, replacing all occurrences in each $U_i$ and $V_i$ by the empty word) the solution~$\sig$ can be made nonsingular. 
This leads to the following nondeterministic procedure: 
we guess a subset $\cX'\sse \cX$, and then we remove all variables in $\cX'$ from the equations and the set of variables; simultaneously, we strengthen the requirement that all \solus\ of the resulting system $E_{\cX'}$ must be nonsingular. 
The procedure is sound and complete:
we have $\Sol(E)\neq \es$ (resp.~$\abs{\Sol(E)} = \infty$) \IFF $\Sol(E_{\cX'})\neq \es$ (resp.~$\abs{\Sol(E_{\cX'})} = \infty$) holds for one of these guesses $E_{\cX'}$.
Moreover, $E$ has infinite exponent of periodicity (according to \prref{def:defexpinf} below) \IFF one of these systems does. 
Thus, for describing $\Sol(E)$ it is enough to work with nonsingular 
\solu{s} or, alternatively, with systems of word equations over free semigroups. 

\subparagraph*{Word Equations over Free Semigroups.}
A \emph{system of word equations with regular constraints} over $\OO^+ = (\Sig\cup\cX)^+$ is a pair $(\cS, \mu)$ where $\cS \sse \OO^+ \times \OO^+$ is a finite set of \emph{word equations}, and where $\mu: \OO^+ \to S$ is a \morph to a finite semigroup~$S$. 
A \emph{\solu} of $E=(\cS,\mu)$ is given by a mapping~$\sig: \cX \to \Sig^+$ with the same properties (1) and (2) as in the monoid case.

The semigroup setting has the advantage that we can use constraints in finite semigroups without neutral elements, allowing us to state some of our results in a more natural way.\footnote{In some sense, the theory of finite semigroups is indeed richer than the one of finite monoids.}

\smallskip

In order to reduce equations over monoids to equations over semigroups, we have already implicitly used one of the central tools in the study of word equations, namely the following.
\begin{definition}\label{def:defsubst}
A \emph{\subst} for $\OO=\Sig\cup\cX$ is an \Endo~$\tau$ of $\OO^+$ leaving the constants in $\Sig$ invariant, and therefore
defined by its restriction to $\cX$.
We say that a \subst is \emph{defined by} a mapping $X \mapsto \tau(X)$  for some $X\in \cX$ if $\tau(\alp) = \alp$ for all $\alp \in \OO \sm \os{X}$.
A \emph{basic \subst} is a \subst defined by 
$X \mapsto \tau(X)$ with $\tau(X) \in \OO X\cup \Sig$.
\hspace*{\fill}$\diamond$\end{definition}

Note that (compositions of) basic \substs are powerful enough to compute every \solu~$\sig$, but \textit{cave canem}: in general, basic \subst{s} increase the length of the equation.

\subsection{Reducing Systems of Equations}\label{sec:single}

Frequently it is convenient to encode a system of word equations $\cS = \os{(U_1,V_1)\lds (U_n,V_n)}$ with constraint $\mu \colon \OO^+ \to S$ by a single word equation. 
In order to do so, we first extend the set of constants $\Sig$ by introducing a new constant $\# \notin \Sig$.
Let $\tilde \Sig = \Sig \cup \os \#$ be this new set of constants, and let $\tilde \OO = \tilde \Sig \cup \cX$.
By adjoining a new zero element $0$ to $S$, we obtain a semigroup $\tilde S = S \cup \os 0$ and we extend $\mu$ to a homomorphism $\tilde \mu : \tilde \OO^+ \to \tilde S$ with $\tilde \mu(\#) = 0$.
Having this, we finally replace $(\cS,\mu)$ by $((U_1\# \cdots U_n\#\, ,V_1\# \cdots V_n\#), \tilde\mu)$, which has the same solutions.

Note that the above encoding does not create any new variables, and it makes sure that the total number of occurrences of each variable is the same as before.
Thus, this construction allows us to subsequently focus on the case of a single word equation.

Henceforth, we also write a word equation $(U,V)$ in the more readable form as $U=V$.

\subsection{The Exponent of Periodicity}\label{sec:expo}
The exponent of periodicity is a well-known concept in combinatorics on words.
Formally, the \emph{exponent of periodicity} $\exp(w)$ of a word $w\in\Sig^*$, is the greatest number such that there is a factorization $w=u\, p^{\exp(w)}\, v$ with $p\in \Sig^+$.
Note that $\exp(w)=0$ \IFF $w = 1$. 

\begin{definition}\label{def:defexpinf}
Let $E = (U=V,\mu)$ be a word equation (over a free monoid or semigroup) with a regular constraint~$\mu$. 
The \emph{exponents of periodicity} of a \solu~$\sig$ of $E$, which we denote by $\exp(\sig)$, and of $E$ itself, denoted by $\exp(E) = \exp(U=V, \mu)$, are defined as follows. 
\begin{itemize}
\item We let $\exp(\sig)$ be the maximum $\max\set{\exp(\sig(X))}{X\in \cX}\in \N$.
\item We let $\exp(E)$ be the supremum $\sup\set{\exp(\sig)}{\sig\in \Sol(E)}\in \N\cup \os \infty$. \hspace*{\fill}$\diamond$
\end{itemize}
\end{definition}

The next proposition generalizes a classical result of  Hmelevski\u{\i}~\cite{hme76}. 
  His proof is based on $p$-stable normal forms, which were also used by Makanin~\cite{mak77} and 
Schulz~\cite{schulz90} for solving word equations with regular constraints.
The description of Makanin's algorithm in~\cite{die98lothaire} contains a proof of \prref{prop:hme76} with explicit bounds, even in case~$S$ is part of the input.

\begin{proposition}\label{prop:hme76}
Let $E = (U=V,\mu)$ be a word equation with a regular constraint~$\mu$. 
Then there is a computable $\kappa_E\in \N$ such that we have $\exp(E) > \kappa_E$ \IFF there are $X \in \cX$ and a nonempty word $p\in \Sig^+$ such that for each $n \in \N$ there is a \solu $\sig_n$ of $E$ where $p^n$ appears as a factor in $\sig_n(X)$. 
\Ip if $\exp(E) > \kappa_E$, then $\exp(E) = \infty$ and $\abs{\Sol(E)} = \infty$.
Moreover, if $\mu$ maps to a fixed finite semigroup~$S$, then we can choose $\kappa_E \in 2^{\Oh(|UV|\cdot |\OO|)}$.
\end{proposition}

The above result has the following noteworthy application.

\begin{corollary}[Plandowski, Schubert~\cite{PlandowskiS2018tcs}]\label{cor:ps19}
Let $E = (\cS, \mu)$ be any system of word equations with a regular constraint $\mu: \OO^+ \to S$. 
If the finite semigroup $S$ is fixed and not part of the input, then there is a $\P$-reduction of the problem to decide on input of $E$ whether $\exp(E) = \infty$ to the satisfiability problem of word equations, that is, to decide on input of $E$ whether there exists a \solu\ of $E$.
\end{corollary}
\begin{proof} 
As explained in \prref{sec:single}, we can encode a system by a single equation $(U = V,\mu)$.
It is shown in \cite{die98lothaire} that that we can choose a concrete value $m \in \Oh(|UV|\cdot |\OO|)$ so that $\kappa_E < 2^m$ holds for the number $\kappa_E$ in \prref{prop:hme76}.
Knowing $m$, we consider the variables $X\in \cX$ one after another. 
For each $X \in \cX$, we create a system of word equations with fresh variables $X_0 \lds X_m, Y, Z$, and additional equations $X = Y X_0 Z$ and $X_{i-1} = X_i X_i$ for $1 \leq i \leq m$.

Having this new system we decide whether it has a \solu. 
If the answer is negative for every variable $X \in \cX$, then $\exp(E) < \infty$. 
Otherwise, $\exp(E) = \infty$ by \prref{prop:hme76}.\footnote{The standard way to extend the result to free monoids is by guessing a constant $a$ and introducing a new equation $X_0=aX'_0$ where $X'_0$ is another fresh variable.}
\end{proof}

\subsection{Two-Variable Word Equations}\label{sec:TVW}

A \emph{two-variable word equation} is a word equation $U = V$ over a set of variables $\cX = \os{X,Y}$. 
They were studied deeply in \eg~\cite{IliePlandoswskiRAIRO2000} and further in 
\cite{DabrowskiPlandowskiICALP2004}. 
The main structural result about the full \solu set for two-variable word equations is Theorem~36 in \cite{IliePlandoswskiRAIRO2000}. 
Its proof is a careful case analysis. 
Inspecting these cases shows the following.

\begin{proposition}[Ilie, Plandowski]
\label{prop:IPR2000}
If a two-variable word equation $U = V$ has infinitely many \solus, then $\exp(U=V) =\infty$. 
\end{proposition}

\section{Nice Semigroups}\label{sec:NSG}

As mentioned in the introduction, we are interested in when $|\Sol(U = V, \mu)| = \infty$ implies an infinite exponent of periodicity for a given word equation with constraints $(U = V, \mu)$.
This naturally leads to the following definition.

\begin{definition}\label{def:nicecon}
A finite semigroup~$S$ is \emph{nice} for a family of word equations~$\cE$, if the following condition holds for all equations $(U = V) \in \cE$ over $\OO^+ = (\Sig\cup \cX)^+$ and constraints $\mu: \OO^+ \to S$:
if $(U = V,\mu)$ has infinitely many \solu{s}, then $\exp(U = V,\mu)=\infty$.\hspace*{\fill}$\diamond$
\end{definition}

An immediate consequence of this definition is that a finite semigroup $S$ is nice for the family of word equations $\cE$ if and only if its left-right dual $\Sop$ is nice for the family $\cE^{\mathrm{rev}}$ consisting of all reversed equations $(U^{\mathrm{rev}} = V^{\mathrm{rev}})$ with $(U = V) \in \cE$.
In particular, if $\cE = \cE^{\mathrm{rev}}$, as for quadratic or two-variable equations, then $S$ is nice for $\cE$ if and only if $\Sop$ is.

A reformulation of our conjecture in \prref{sec:bcon} is the following.
\begin{conjecture} \label{conj:watnu}
The trivial semigroup is nice for the class of all word equations.
\end{conjecture}

\subsection{Divisors of Nice Semigroups}

The following oberservations, \prref{lem:niceS} and \prref{lem:niceH}, show
that the class of all finite semigroups which are nice for any given family of word equations $\cE$ is closed under taking \emph{divisors} (that is, homomorphic images of subsemigroups).

\begin{lemma}\label{lem:niceS}
  If $S$ is a nice semigroup for $\cE$, then so is every subsemigroup $S'$ of $S$.
\end{lemma}
\begin{proof}
  Let $S'$ be a subsemigroup of $S$, and let $\iota: S' \to S$ be the corresponding inclusion.
  Then, for all word equations $(U = V)$ over $\OO^+ = (\Sig\cup \cX)^+$ and all constraints $\mu': \OO^+ \to S'$, the solution sets $\Sol(U = V,\mu')$ and $\Sol(U = V, \iota\mu')$ coincide.
  Hence, $(U = V,\mu')$ has infinitely many solutions (resp.\ infinite exponent of periodicity) \IFF $(U = V, \iota \mu')$ does.
\end{proof}

\begin{lemma}\label{lem:niceH}
  If $S$ is a nice semigroup for $\cE$, then so is every homomorphic image $S'$ of $S$.
\end{lemma}
\begin{proof}
  Let $\rho: S \to S'$ be an epimorphism.
  Further, suppose that $(U = V) \in \cE$ is a word equation over $\OO^+ = (\Sig\cup \cX)^+$ such that $\abs{\Sol(U = V,\mu')} = \infty$ where $\mu': \Omega^+ \to S'$.
  Let us first note that for every $\sig \in \Sol(U = V,\mu')$ there is some constraint $\mu_\sig: \Omega^+ \to S$ with $\mu' = \rho\mu_\sig$ such that $\sig \in \Sol(U = V,\mu_\sig)$.
  Indeed, choosing $\mu_\sig(a) \in \oi\rho(\mu'(a))$ arbitrarily for all $a \in \Sig$ and setting $\mu_\sig(X) = \mu_\sig(a_1)\mu_\sig(a_2) \cdots \mu_\sig(a_n)$ where $\sig(X) = a_1a_2 \cdots a_n$ with $a_i \in \Sig$ for all $X \in \cX$ defines such a constraint $\mu_\sig$.
  Since there exist only a finite number of constraints $\mu: \OO^+ \to S$, we must have $\abs{\Sol(U = V,\mu)} = \infty$ for some $\mu$ with $\rho\mu = \mu'$ by the pigeonhole principle.

  Assuming that $S$ is nice for~$\cE$, we deduce that $\exp{(U = V,\mu)} = \infty$.
  In turn, this implies $\exp{(U = V,\mu')} = \infty$, as required, since $\Sol(U = V,\mu)$ is contained in $\Sol(U = V,\mu')$.
\end{proof}

\subsection{Extensions of Nice Semigroups}\label{sec:nilpot}

An \emph{extension} of a semigroup $S$ is a semigroup $T$ that contains $S$ as an ideal.
Such an extension is \emph{retractive} if there exists a homomorphism $\rho: T \to S$ with $\rho(x) = x$ for all $x \in S$; that is, if $S$ is a retract of $T$.
It is \emph{nilpotent} if $T^k \sse S$ for some $k$ or, equivalently, if the Rees quotient $T/S$ (in which all elements of $S$ are identified) is a nilpotent semigroup.

Recall the notion of (basic) \subst in \prref{def:defsubst}.
A \subst~$\tau$, thus an \Endo of $\OO^+$ leaving $\Sig$ invariant, is called \emph{trivial} if for all $X\in \cX$ we have 
$\tau(X)\in \Sig^*X \cup \Sig^+$. Note that every trivial \subst
can be written as a sequence of basic \substs.

\begin{theorem}\label{thm:ell}
  Let $\cE$ be a family of word equations which is closed under trivial \subst{s}.
  Further, let $T$ be a finite semigroup which is a retractive nilpotent extension of a semigroup~$S$.
  Then $S$ is a nice semigroup for $\cE$ \IFF $T$ is a nice semigroup for $\cE$.
\end{theorem}
\begin{proof}
If $T$ is nice for $\cE$, then its subsemigroup (and homomorphic image) $S$ is nice for $\cE$ by \prref{lem:niceS} (or \prref{lem:niceH}). 
Thus, it suffices to show that $T$ is nice for $\cE$ \mbox{assuming that $S$ is so.}

Suppose that the word equation $(U = V) \in \cE$ with constraints $\mu: \OO^+ \to T$ has infinitely many solutions; that is, $\abs{\Sol(U = V,\mu)} = \infty$.
We will show that this implies $\exp(U = V,\mu) = \infty$.

Let us call a variable $X$ \emph{inessential} if, for some fixed $w \in \Sigma^+$, the equation~$E = (U = V,\mu)$ has infinitely many solutions $\sigma$ with $\sigma(X) = w$.
Note that we can simply eliminate such a variable $X$ by following the trivial substitution defined by $X \mapsto w$, since $E$ has infinite exponent of periodicity whenever the new word equation does so.
Hence, we may as well assume that none of the variables is inessential.
For every $k \in \N$, all but finitely many of the solutions $\sigma$ of the word equation $E$ must therefore satisfy $\abs{\sigma(X)} > k$ for all $X \in \cX$.

Next, let us fix a number $k$ large enough so that $T^{k+1} \sse S$.
From the above we see that $\mu(X) \in S$ for every $X \in \cX$.
Let us also fix a retraction $\rho \colon T \to S$, and consider the word equation $E_\rho = (U=V, \rho\mu)$ with constraints in $S$.
As every solution of $E$ is a solution of $E_\rho$, the latter has infinitely many solutions $\sigma$ with $\abs{\sigma(X)} > k$ for all $X \in \cX$.
Note that, conversely, every solution $\sigma$ of $E_\rho$ with $\abs{\sigma(X)} > k$ for all $X \in \cX$ is also a solution of $E$.
We claim that there are such solutions with arbitrarily large exponent of periodicity.

Indeed, by the pigeonhole principle, there exist words $w_X \in \Sigma^k$ such that $\sigma(X) \in w_X \Sigma^+$ holds for all $X \in \cX$ for infinitely many solutions $\sigma$.
Hence, some $E'_\rho = (U' = V', \mu') \in \cE$ with constraints in $S$ and infinitely many solutions can be obtained from $E_\rho$ by the trivial substitution $X \mapsto w_X X$.
By assumption on $S$ there are solutions $\sigma'$ of $E'_\rho$ with arbitrarily large exponent of periodicity.
The corresponding solutions $\sigma$ of $E_\rho$ with $\sigma(X) = w_X \sigma'(X)$ have the desired property since $\abs{\sigma(X)} = \abs{w_X} + \abs{\sigma'(X)} \geq k + 1$ holds for all $X \in \cX$.
\end{proof}

In order to apply \prref{thm:ell} we need at least one nonempty finite semigroup $S$ which is nice for a family of word equations $\cE$. 
By \prref{lem:niceH} (or \prref{lem:niceS}) the trivial semigroup would then be nice for $\cE$. 
Hence,  
if \prref{conj:watnu} holds for some family of equations $\cE$ (without constraints), then all finite nilpotent semigroups are nice for this family~$\cE$.
A first example are two-variable word equations, which we briefly discussed in \prref{sec:TVW}. 
Based on \prref{prop:IPR2000} we can state a more general form of this proposition by using \prref{thm:ell}.

\begin{corollary}\label{cor:IPR2000}
  Every finite nilpotent semigroup is nice for two-variable word equations.
\end{corollary}

Unfortunately we were not able to prove that the trivial semigroup is nice for the class of all word equations over free semigroups.

\section{Quadratic Word Equations}\label{sec:freeQWE}
A system $\cS=\os{(U_1,V_1)\lds (U_n,V_n)}$ of word equations is called \emph{quadratic} if every variable occurs at most twice in its equations, that is, $\abs{U_1\cdots U_nV_1\cdots V_n}_X \leq 2$ holds for all $X \in \cX$.
As explained in \prref{sec:single}, it will suffice for our purposes to deal with (single) quadratic word equations $U = V$. 
Throughout, we denote the class of all quadratic word equations by $\cQE$.

The following \prref{prop:trivial} states that the trivial semigroup is nice for the family of quadratic word equations. 
This result was initially worked out in 2019 by Bastien Laboureix in his (unpublished) stage report~\cite{Laboureix2019} for the Laboratoire Spécification et Vérification (now at the ENS Paris-Saclay, France) when he was on leave in Stuttgart as part of a collaboration between Benedikt Bollig and the first author.
The approach of Laboureix was based on ideas of Olga Kharlampovich (personal communication). 
The proof of the following proposition, however, is from the bachelor's thesis of the second author \cite{natterer2024quadratic}. 
While this proposition is a special case of \prref{thm:mainiac}, we give a simpler and more direct proof here as a warm-up, as this proof also contains some of the main ideas used later in a more general setting.

Since, for every semigroup, there is only one \hom $\mu$ to the trivial semigroup~$\os 1$, we suppress the specification of $\mu$ and speak about 
\emph{word equations without constraints}.

\begin{proposition}\label{prop:trivial}
The trivial semigroup is nice for the family of quadratic word equations. 
Recall, this means that every quadratic word equation over free semigroups without constraints has either only finitely many \solus or it has an infinite exponent of periodicity.
\end{proposition}
\begin{proof}
Let $U=V$ be a quadratic equation in a free semigroup and, as usual, $\OO$ be the disjoint union of the set of constants $\Sig$ and the set of variables $\cX$.
We distinguish between several types of variables. 
The set of variables~$X\in \cX$ which appear~$i$~times in~$UV$ is denoted by $\cX_i$; thus, $\cX = \cX_0 \cup \cX_1 \cup \cX_2$.
Moreover, we say that a variable $X\in \cX_2$ is \emph{balanced} if it appears on both sides of the equation, that is, $|U|_X=|V|_X=1$.
Otherwise, $X \in \cX_2$ is called \emph{unbalanced}.

The proof is by induction, primarily on the number of variables and, secondarily, on the length~$\abs{UV}$.
We assume without restriction that the equation~$U = V$ has infinitely many solutions.
Thus, $\Sig\neq \es \neq \cX$. 
We begin by dealing with the following simple cases.
\begin{itemize}
  \item Suppose that $\cX_0 \neq \es$.
Then we fix some solution $\sig$ of the equation~$U = V$ and some~$X \in \cX_0$.
Let $a \in \Sigma$ be any letter.
For all $n \in \N$ there is a solution $\sig_n$ with $\sig_n(X) = a^n$ and $\sig_n(Y) = \sig(Y)$ for all $Y \in \cX\sm \os X$. Clearly, $\exp(\sig_n) \geq n$. 
Hence, we can assume $\cX_0 = \es$.

\item Suppose that both sides of the equation $U = V$ begin with a constant.
  Then they must begin with the same constant $a \in \Sigma$, for otherwise the equation has no solution.
  We then have $U = aU'$ and $V = aV'$ with $U', V' \in \Omega^+$, for if $U' = 1$ or $V' = 1$, then the equation has at most one solution (since $\cX_0 = \es$).
We then continue with the equation $U' = V'$, which has the same solutions but is shorter.
So, we are done by induction.

\item Suppose that both sides of the equation $U = V$ begin with the same variable; that is, we have $U = XU'$ and $V = XV'$ for some $X \in \cX$ and $U', V' \in \Omega^\ast$. 
Then we may simply remove the variable $X$ from the equation, leaving us in a situation where $\cX_0 \neq \es$.
Henceforth, we therefore also assume that $U$ and $V$ do not begin with the same variable.

\item Suppose that one side of the equation begins with a balanced variable.
Up to symmetry, we then have $U = XU'$ and $V = V'_1XV'_2$ for some $X \in \cX$, $V'_1 \in \Omega^+$, and $U', V'_2 \in \Omega^*$.
Let us then fix some solution $\sig$ out of the infinitely many ones at our disposal.
Observe that the word $v = \sig(V'_1)$ is nonempty since $V'_1 \neq 1$.
Thus, we obtain an infinite family of solutions~$\sig_n$ defined by $\sig_n(X) = v^n \sig(X)$ and $\sig_n(Y) = \sig(Y)$ for all $Y \in \cX$ with $Y \neq X$, and these solutions clearly satisfy $\exp(\sig_n) \geq n$.
\end{itemize}
Our goal is now to reduce the general situation to one of the simple cases above.
We do so by repeatedly applying a suitable substitution $\tau \colon \Omega^+ \to \Omega^+$, thereby transforming the equation $U = V$ into a new equation $\tau(U) = \tau(V)$, or an equivalent one, that still has infinitely many solutions.
This relies on the fact that every solution $\sigma'$ of $\tau(U) = \tau(V)$ gives rise to a corresponding solution $\sigma = \sigma' \tau$ of the original equation $U = V$, and these solutions satisfy $\exp(\sigma) \geq \exp(\sigma')$.
Using this method, we can handle the following additional simple case.
\begin{itemize}
\item Suppose that one side of the equation $U = V$ consists of a single variable.
Up to symmetry, we may then assume $U = X$ for some $X \in \cX$.
By the above cases, we may further assume that $\abs{V}_X = 0$.
The substitution $\tau$ defined by $X \mapsto V$ then turns the equation into the trivial but equivalent equation $V = V$.
Since this substitution eliminates the variable $X$, we are done by induction on the number of variables.
\end{itemize}
From now on, one side of the equation begins with a variable.
If exactly one side begins with a variable, then, up to symmetry, we have $U = XU'$ and $V = aV'$ for some $X \in \cX$, $a \in \Sig$ and $U', V' \in \Omega^\ast$.
We may assume that $U', V' \in \Omega^+$, since $U'=1$ is the simple case we just discussed, and $V' = 1$ implies that the equation has at most one solution.

If the substitution $\tau$ defined by $X \mapsto a$ leads to an equation $\tau(U') = \tau(V')$ with infinitely many solutions, then we are done by induction on the length of the equation.
Otherwise, we are forced to continue with the equation $X \tau(U') = \tau(V')$ where $\tau$ is the substitution defined by $X \mapsto aX$.
Since we may assume that $X$ is not balanced, we have $\tau(V') = V'$, which is shorter than $V = aV'$.
Hence, by repeating this step if necessary, we eventually
arrive at the situation where both sides of the equation begin with a variable; that is, we have $U = XU'$ and $V = YV'$ with $X,Y\in \cX$ and $X \neq Y$.
Note that we can assume $U', V' \in \Omega^+$.

If the substitution $\tau$ defined by $X \mapsto Y$ leads to an equation $\tau(U') = \tau(V')$ with infinitely many solutions, then we are done by induction on the number of variables.
Otherwise, we are forced to continue with a substitution defined by $X \mapsto YX$ or by $Y \mapsto XY$.

We make a case distinction. 
Firstly, let $X \in \cX_1$ and suppose that we use a substitution~$\tau$ defined by $X\mapsto YX$.
This leads to the new equation $XU' = V'$.
This equation is shorter and, hence, if $XU' = V'$ still has infinitely many solutions, then we are done by induction.

This means that we can treat variables in $\cX_1$ as if they were constants in the following sense. If we see an equation $XU' = YV'$ where $X\in \cX_2$ and $Y \in \cX_1$, 
then we continue with the \subst defined by $X\mapsto YX$ because
the other choice $Y\mapsto XY$ leads to a shorter equation.
Thus, the (only) case of interest is an equation $XU' = YV'$ where $X,Y\in \cX_2$ are unbalanced.
By symmetry we continue using a substitution defined by $X\mapsto YX$, which then leads to $XU_1YXU_2 = V_1YV_2$. 
Note that this makes $Y$ balanced.
The crucial observation is that, in the described process, it is impossible for $Y$ to become unbalanced again: this could only happen if one of the sides begins with the now balanced $Y$, but at this point the process stops. 
Thus, this argument bounds the number of successive equations where both sides start with an unbalanced variable by $|\cX|$, completing the proof.
\end{proof}

An immediate application of \prref{thm:ell} yields 
the following.

\begin{corollary}\label{cor:proptrivial}
Every finite nilpotent semigroup is nice for the family $\cQE$. 
\end{corollary}

\subsection{The Matiyasevich Solution Graph}
In 1968, \Mati~\cite{Matiyasevich68} defined a \solu graph of linear size in the input for showing that the satisfiability problem for quadratic equations without constraints and without \invol is decidable in $\PSPACE$ for free semigroups and for free monoids. 
The same approach copes with an \invol and regular constraints. It also provides a simple method to decide whether or not a quadratic equation has infinitely many \solu{s}.

Let $n\in \N$ and let $S$ be a fixed semigroup. 
The \emph{\MatiSG} for the given $n$ is defined as an edge-labeled directed graph $\cA_n$. The vertices are called \emph{states} and the edges, labeled by basic \subst{s} according to \prref{def:defsubst}, are called \emph{transitions}.

Each state is a triple $(U = V,\cX,\mu)$ with $|UV| \leq n$ and $|\cX| \leq n$.
The state is called \emph{final}, if $U = a = V$ for some $a \in \Sigma$ and if $\cX = \emptyset$. Next we describe the transitions.

We start with a state of the form $(\alpha U = \alpha V, \cX, \mu)$ with $\alpha \in \Omega$. We then have a unique outgoing \epstra, labeled by the identity, to the state $(U = V,\cX,\mu)$.

Every other state $(U = V,\cX,\mu)$ always has the following two outgoing transitions for every $X \in \cX$ which does not appear in $UV$ and for every $a \in \Sigma$.
\begin{itemize}
    \item A transition labeled $\tau(X) = aX$ to the state $(U = V,\cX,\mu')$.
    \item A transition labeled $\tau(X) = a$ to the state $(U = V,\cX\setminus\{X\}, \mu')$.
\end{itemize}
Moreover, if a state can be written as $(XU = \alpha V, \cX, \mu)$ then we have the following two additional outgoing transitions.
\begin{itemize}
    \item A transition labeled $\tau(X) = \alpha X$ to the state $(X\tau(U) = \tau(V), \cX, \mu')$.
    \item A transition labeled $\tau(X) = \alpha$ to the state $(\tau(U) = \tau(V), \cX\setminus\{X\}, \mu')$.
    \item There are also symmetric
  rules for $\alpha V=XU$.

\end{itemize}
The constraints $\mu'$ are always chosen such that $\mu = \mu'\tau$.
If this is not possible, the corresponding transition does not exist.
If $\mu'$ can be chosen in several different ways, then multiple copies of the respective transition are introduced.

Given a quadratic equation $(U = V,\mu)$ with constraints and $|UV|\leq n$, we let $\cX$ be the set of variables occurring in $UV$ and set $E = (U = V,\cX,\mu)$.
Note that $E$ is a state in $\cA_n$.
We can construct an NFA $\cA(E)$ from $\cA_n$ as follows. 
First, we declare $E$ to be the unique initial state, keeping the final states already defined, and then we trim the automaton accordingly.
The resulting automaton $\cA(E)$ is called the \emph{(Matiyasevitch) solution graph} of~$E$.

\begin{figure}[ht]
\begin{center}
\begin{tikzpicture}[yscale=0.18, xscale=0.46]
\node (1) at (10, 20) {$\mathbf{XabY = YbaX}$};

\node (2a) at (0, 10) {$\mathbf{XabY = baYX}$};
\node (3a) at (0, 20) {$\mathbf{XabY = aYbX}$};
\node (4a) at (0, 30) {$abY = Yba$};
\node (5a) at (0, 40) {$baY = Yba$};
\node (6a) at (0, 0) {$bY = Yb$};
\node (7a) at (10, 0) {$b = b$};

\node (2b) at (20, 30) {$\mathbf{abXY = YbaX}$};
\node (3b) at (20, 20) {$\mathbf{bXaY = YbaX}$};
\node (4b) at (20, 10) {$Xab = baX$};
\node (5b) at (20, 0) {$Xab = abX$};
\node (6b) at (20, 40) {$Xa = aX$};
\node (7b) at (10, 40) {$a = a$};

\draw[->,thick] (1) edge node[above left, scale=0.8] {$Y \mapsto XY$} (2b);
\draw[->,thick] (2b) edge node[right, scale=0.8] {$Y \mapsto aY$} (3b);
\draw[->,thick] (3b) edge node[above, scale=0.8] {$Y \mapsto bY$} (1);
\draw[->,thick] (3b) edge node[right, scale=0.8] {$Y \mapsto b$} (4b);
\draw[->,thick] (4b) edge[bend left=25] node[right, scale=0.8] {$X \mapsto bX$} (5b);
\draw[->,thick] (5b) edge[bend left=25] node[left, scale=0.8] {$X \mapsto aX$} (4b);
\draw[->,thick] (2b) edge node[right, scale=0.8] {$Y \mapsto a$} (6b);
\draw[->,thick] (6b) edge[loop above, looseness=50] node[above, scale=0.8] {$X \mapsto aX$} (6b);
\draw[->,thick] (6b) edge node[above, scale=0.8] {$X \mapsto a$} (7b);
\draw[->,thick] (4b) edge[bend right=40] node[above left, scale=0.8] {$X \mapsto b$} (7a);

\draw[->,thick] (1) edge node[below right, scale=0.8] {$X \mapsto YX$} (2a);
\draw[->,thick] (2a) edge node[left, scale=0.8] {$X \mapsto bX$} (3a);
\draw[->,thick] (3a) edge node[below, scale=0.8] {$X \mapsto aX$} (1);
\draw[->,thick] (3a) edge node[left, scale=0.8] {$X \mapsto a$} (4a);
\draw[->,thick] (4a) edge[bend left=25] node[left, scale=0.8] {$Y \mapsto aY$} (5a);
\draw[->,thick] (5a) edge[bend left=25] node[right, scale=0.8] {$Y \mapsto bY$} (4a);
\draw[->,thick] (2a) edge node[left, scale=0.8] {$X \mapsto b$} (6a);
\draw[->,thick] (6a) edge[loop below, looseness=50] node[below, scale=0.8] {$Y \mapsto bY$} (6a);
\draw[->,thick] (6a) edge node[below, scale=0.8] {$Y \mapsto b$} (7a);
\draw[->,thick] (4a) edge[bend right=40] node[below right, scale=0.8] {$Y \mapsto a$} (7b);
\end{tikzpicture}
\caption{The \Mati solution graph of the equation $XabY = YbaX$ without constraints.}
\label{fig:XabX}
\end{center}
\end{figure}
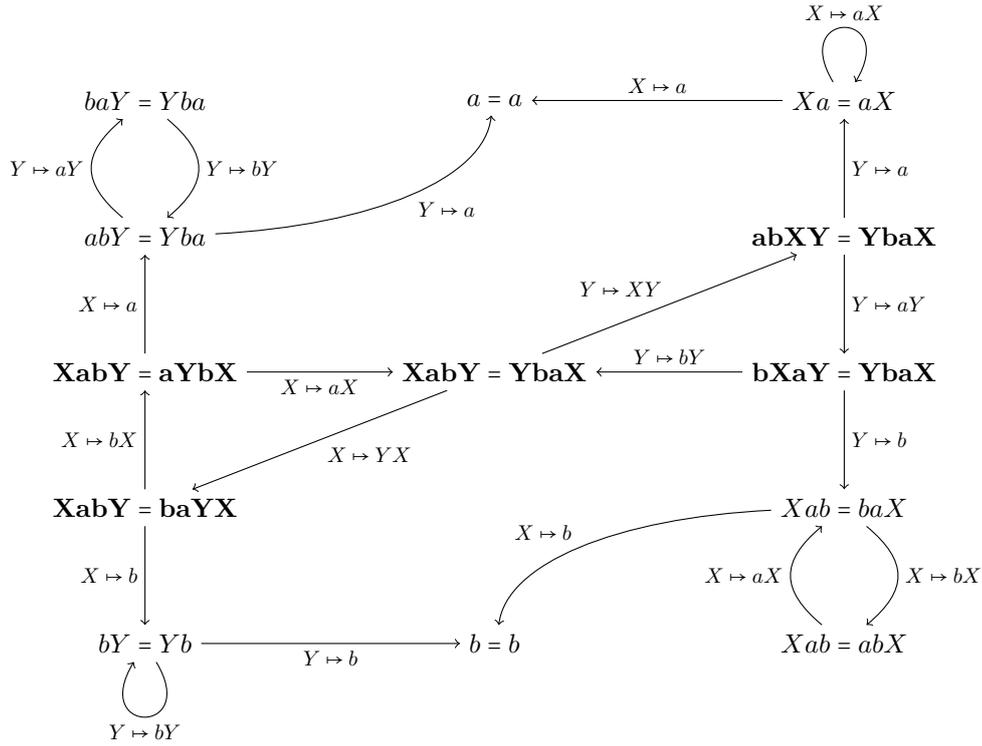

\begin{example}
    The diagram in \prref{fig:XabX} shows the solution graph of the quadratic equation $XabY = YbaX$ without constraints where some \epstra{s} are omitted for clarity.
    Unsolvable states do not appear as the automaton is trim. The equations in boldface form the strongly connected component where both $X$ and $Y$ appear in the equation. 
\end{example}
\begin{proposition}\label{prop:matiya}
  On input $n\in \N$, the \MatiSG $\cA_n$ can be effectively constructed.
Moreover, for every state $E$ of the automaton $\cA_n$, the following 
properties hold.
\begin{itemize}
  \item The trim subautomaton $\cA(E)$ accepts the language $\Sol(E)$ (where composition of endo\-morphisms is defined in reverse order).
\item We have $|\Sol(E)|=\infty$
\IFF $\cA(E)$ contains some nontrivial cycle.
\end{itemize}
\end{proposition}
\begin{proof}
That $\cA_n$ can be effectively constructed is obvious.
Now, let $E$ be a state of $\cA_n$.

If $\pi$ is an accepting path in $\cA(E)$, then $\pi$ is labeled 
by an \Endo\ $\tau_m\cdots \tau_1$ where each $\tau_i$ is a \subst, and
defines a \solu~$\sig$ with 
$\sig(X)= \tau_m\cdots \tau_1(X)$ for~all~$X\in \cX$.

By induction on the length of a \solu, it is easy to see that 
 every \solu has such a factorization. 
Since $\cA(E)$ is finite, we obtain that $|\Sol(E)|=\infty$
\IFF there is an accepting path using a \tra{} within some strongly connected component of $\cA(E)$.
\end{proof}
For example, \prref{fig:XabX} shows that the equation $XabY = YbaX$ has infinitely many \solu{s}.

\begin{remark}\label{rem:genWE}
Both problems related to the questions in the next \prref{cor:matiya} remain decidable for 
all word equations with regular constraints in free monoids or semigroups (with or without \invol) and extend to free groups with rational constraints in \cite{DiekertJP16}~and~\cite{CiobanuDiekertElder2016ijac}.
There are also extensions to virtually free groups \cite{DiekertElder2020ijac}
and, even further, to hyperbolic groups~\cite{CiobanuElder2021}. 
All these papers are quite technical. 
Dealing with quadratic equations makes life nicer, as the proof follows directly from properties of the \MatiSG $\cA_n$.
\hspace*{\fill}$\diamond$\end{remark}

The next corollary is a direct consequence of \prref{prop:matiya} since we can encode a system of quadratic word equations as a single quadratic equation as in \prref{sec:single}. 

\begin{corollary}\label{cor:matiya}
On input of a system $\cS$ of quadratic word equations over $\OO=\Sig \cup \cX$, it is decidable whether $(\cS,\mu)$ has some \solu, and whether $(\cS,\mu)$ has infinitely many \solu{s}.
\end{corollary}

\section{The Main Result}\label{sec:sobeauty}
Using the Matiyasevitch solution graph defined in \prref{sec:freeQWE}, we will prove our main result on the exponent of periodicity of quadratic equations with certain constraints (which are defined in \prref{sec:lst}, and discussed in more detail in \prref{sec:ddlg}).
To do so, we first introduce a special type of equation for which the desired claim is particularly easy to establish.
The general case of quadratic equations is then reduced to this special case.
In order to achieve this reduction, we establish a set of invariants of the strongly connected components of the corresponding solution graph.

\subsection{\texorpdfstring{$\cL$-Stabilizers}{L-Stabilizers}}\label{sec:lst}

The following notation plays a key role in our construction.

\begin{definition}\label{def:lst}
  Let $S$ be a finite semigroup.
  We say that an element $u \in S^1$ is an \emph{$\cL$-stabilizer} of an element $x \in S$ provided that $u^{\oo}x = x$.
  We denote the set of all such elements by 
  \[
    \stab_{\cL}(x) = \set{u \in S^1}{ u^{\oo} x = x}
  . \hspace*{\fill} \diamond \]
\end{definition}

This terminology stems from the observation that the orbit of an element $x$ under the action of multiplication on the left by any power of an $\cL$-stabilizer $u$ of $x$ is completely contained in $\cL(x)$. 
Indeed, if $u^\oo x = x$, then $x = u^\oo x \leL u^n x \leL x$ holds for all $n \in \N$.

Subsequently, we will be interested in the following class of finite semigroups for which the converse of the above implication holds (thereby justifying our terminology) and for which $\cL$-stabilizers are particularly well-behaved.
 As we will see later, this class of semigroups coincides with the variety \pv{DLG} of finite semigroups where each regular $\cD$-class is a right group or, equivalently, each regular $\cL$-class is a group; see \prref{lem:lst2} and \prref{lem:lst1}.
However, the arguments in our proofs rely only on the conditions given in \prref{def:DLRG}.
We therefore adopt these conditions as a preliminary technical definition of the class $\pv{DLG}$.

\begin{definition}\label{def:DLRG}
  We say that a finite semigroup $S$ belongs to the class \pv{DLG} if the following conditions are satisfied for all $x,y \in S$ and $u,v \in S^1$.
\begin{bracketenumerate}
\item If $u^{\oo} x = x$ and $v^{\oo} x = x$, then $(uv)^{\oo} x = x$; that is, $\stab_{\cL}(x)$ is a submonoid of $S^1$.
\item If $x \siJ y$, then $u^{\oo} x = x$ if and only if $u^{\oo}y = y$; that is, $x \siJ y \implies \stab_{\cL}(x) = \stab_{\cL}(y)$.
\item If $ux \siL x$, then we have $u^{\oo}x = x$. 
\hspace*{\fill}$\diamond$
\end{bracketenumerate}
\end{definition}

\subsection{Nicely Balanced Equations}\label{sec:origSN}
In \prref{sec:NSG} we defined when a finite semigroup is nice. 
The counterpart for equations is the following definition, which is not restricted to quadratic word equations.
Here, given a \hom $\mu:\OO^+\to S$ to a finite semigroup $S$ and a symbol $\alpha \in \OO$, we write 
\begin{align}\label{eq:lmux}
  L_{\mu,\alpha} = \set{w \in \Sigma^+}{\mu(w) = \mu(\alpha)}.
\end{align}

\begin{definition} \label{def:nice}
Let $U=V$ be an equation with regular constraints $\mu$ and 
with a set of variables~$\cX$. 
It is called \emph{nicely balanced with regard to $X \in \cX$} if, firstly, $\abs{L_{\mu,X}} = \infty$ and, secondly, up to symmetry in $U$ and $V$ one of the following conditions holds.
\begin{itemize}
\item We have $|UV|_X=0$. That is, $X$ does not appear in $U = V$.
\item We have $U = Xu$ and $V = vXv'$ with $u,v,v' \in (\Omega\sm\os X)^*$ and $\mu(v) \in \stab_{\cL}(\mu(X))$.
\hspace*{\fill}$\diamond$
\end{itemize}
\end{definition}

\begin{lemma} \label{lem:nice}
  Let $(U = V, \cX, \mu)$ be a solvable word equation which is nicely balanced with regard to $X \in \cX$ (whether or not it is quadratic).
  Then $\exp(U = V, \cX, \mu) = \infty$.
\end{lemma}
\begin{proof}
Starting with some solution $\sig \in \Sol(U = V, \cX, \mu)$, we will create new solutions with high exponents of periodicity by only changing the value of $\sig(X)$.
There are two cases.
\begin{itemize}
\item Suppose that $|UV|_X = 0$.
  Then any word $w \in L_{\mu, X} \sse \Sigma^+$ induces a solution $\sig_w$ by setting $\sig_w(X) = w$ and $\sig_w(Y) = \sig(Y)$ for all $Y \in \cX \sm \os X$.
  The claim now follows from the pumping lemma for regular languages, since $\abs{L_{\mu, X}} = \infty$ by hypothesis.
\item Suppose that $U = Xu$ and $V = vXv'$ with $u, v, v' \in (\Omega\setminus\{X\})^*$ and $\mu(v) \in \stab_{\cL}(\mu(X))$.
  If $v = 1$, then the equation $U = V$ is equivalent to an equation in which $X$ no longer occurs and we are done by the first case.
    Indeed, if $u \neq 1 \neq v'$, then $U = V$ is equivalent to the equation $u = v'$.
  Otherwise, if $u = 1$ or $v' = 1$, then we must have $u = 1 = v'$ in order for $U = V$ to have a solution. 
  In this case, the equation is of the form $X = X$ and, hence, equivalent to any equation of the form $a = a$ with $a \in \Sigma$. 
  If $v \neq 1$, then we fix a number $\oo\in \N_+$ such that $\mu(v)^\oo$ is idempotent. 
For any $m\in \N$, let~$\sig_m$ be defined by $\sig_m(X) = \sig(v^{m\omega}X)$  and $\sig_m(Y) = \sig(Y)$ for all $Y \in \cX \sm \os X$.
Then $\sigma_m$ solves the equation and has high exponent of periodicity due to $\sig(v) \neq 1$.
The constraints are also satisfied since we have $\mu(\sig_m(X)) = \mu(\sig(v))^{m\omega} \mu (\sig(X)) = \mu(v)^{\omega}\mu(X) = \mu(X)$, where the second equality uses that $\sig$ is a solution of $(U = V, \cX, \mu)$ and the third uses $\mu(v) \in \stab_{\cL}(\mu(X))$.
\end{itemize}
In either case, $(U = V, \cX, \mu)$ has solutions with arbitrarily high exponent of periodicity.
\end{proof}

\subsection{Strongly Connected Components and Leaders}\label{sec:SCCCL}
As usual, a \emph{cycle} of a directed graph is a nonempty directed
path starting and ending in the same vertex; and a \emph{strongly connected component} $C$ is an induced subgraph such that between any two vertices $p,q\in C$ there 
is a directed path from~$p$ to~$q$.

Let $\cA(E_0)$ be the \MatiSG for some quadratic equation $E_0$ with constraints in a finite semigroup~$S$. 
In particular, each state in $\cC$ has at least one outgoing transition. 
Note that each state in $\cC$ is solvable since $\cA(E_0)$ is trim by definition.

The following lemma characterizes the transitions within~$\cC$.
\begin{lemma} \label{lem:sntrans}
  Consider a \tra{} from $E = (U = V, \cX, \mu)$ to $E' = (U' = V', \cX', \mu')$ within~$\cC$.
  Then the following hold.
\begin{bracketenumerate}
\item The transition is labeled by $X \mapsto \alpha X$ for some $X \in \cX$ and $\alpha \in \Omega$.
\item We have $\mu(X) \leq_\cR \mu(\alpha)$ and, \ip $\mu(X) \leq_\cJ \mu(\alpha)$.
\item We have $\mu(X) \sim_\cL \mu'(X)$ and, \ip $\mu(X) \sim_\cJ \mu'(X)$.
\item We have $|L_{\mu,X}| = \infty$ where $L_{\mu,X}$ is defined by \prref{eq:lmux}.
\end{bracketenumerate}
\end{lemma}
\begin{proof}
  Item $(1)$ follows since all other transitions reduce the size of the equation or of $\cX$.
  Item $(2)$ then follows from $\mu(X) = \mu(\alpha)\mu'(X)$.
  This also shows $\mu(X) \leq_\cL \mu'(X)$ and, since $\mu(Y) \leq_\cL \mu'(Y)$ holds for all $Y \in \cX$ and all transitions within $\cC$, this implies~$(3)$.

  To see the item $(4)$, observe that $L_{\mu,\alpha} \cdot L_{\mu',X} \sse L_{\mu,X}$ and that $L_{\mu,\alpha} \neq \es$ since the state $E$ is solvable.
  Moreover, $L_{\mu', Y} = \os{1} \cdot L_{\mu', Y} = L_{\mu, Y}$ for all $Y \in \cX$ with $Y \neq X$.
  Since the transition is within $\cC$, there is thus a nonempty set $K \sse \Sigma^*$ with $L_{\mu, \alpha} \cdot K \cdot L_{\mu, X} \sse L_{\mu, X}$ and, in particular, some $v \in \Sig^+$ with $v \cdot L_{\mu, X} \sse L_{\mu, X}$.
  As $L_{\mu, X}$ is nonempty, it is thus infinite.
\end{proof}

\begin{lemma} \label{lem:snleader}
  Let $E = (\alpha U = \beta V, \cX, \mu)$ be any equation in~$\cC$ with $\alpha, \beta \in \Omega$.
  Then the minimum of $\cJ$-classes $J_E = \min\{\cJ(\mu(\alpha)), \cJ(\mu(\beta))\}$ is well-defined and independent of the chosen state $E$ in~$\cC$.
  In particular, the \emph{leading $\cJ$-class} $J_\cC = J_E$ of~$\cC$ is well-defined.
\end{lemma}
\begin{proof}
  The minimum $\cJ$-class $J_E = \min\os{\cJ(\mu(\alpha)), \cJ(\mu(\beta))}$ is well-defined since $E$ is~solvable.
  Indeed, for any solution $\sigma$ of $E$, one of the words $\sigma(\alpha)$ and $\sigma(\beta)$ is a prefix of the other and, hence, $\mu(\alpha) = \mu(\sigma(\alpha))$ and $\mu(\beta) = \mu(\sigma(\beta))$ are comparable in $\cR$- and, hence, in $\cJ$-order.

  Since~$\cC$ is nontrivial, there is at least one outgoing transition in $\cC$, which is never an \epstra.
  Moreover, such a transition can change the value of $J_E$ only if it affects $\alpha$~or~$\beta$.
  Therefore, using \prref{lem:sntrans}~$(1)$ and by symmetry, we may assume that $\beta = X \in \cX$ and that~$E$ has an outgoing transition to some state $E'$ in $\cC$ with label $X \mapsto \alpha X$.
  Then \prref{lem:sntrans}~$(2)$ implies that $J_E = \cJ(\mu(X))$.
  For any such transition, we then have $E' = (\tau(U), X\tau(V), \mu')$ so that $J_{E'} \leJ \cJ(\mu'(X)) = \cJ(\mu(X)) = J_E$ by \prref{lem:sntrans}~$(3)$.
  Since we remain within $\cC$, this implies $J_E = J_{E'}$.
  Thus, $J_{\cC}$ is also well-defined.
\end{proof}

\subsection{Playgrounds and Players}

This section makes heavy use of the properties from \prref{def:DLRG}. 
From now on, we thus assume that all constraints $\mu : \OO^+ \to S$ are with respect to semigroups $S \in \pv{DLG}$.

Combining the condition $(2)$ in \prref{def:DLRG} with \prref{lem:snleader} allows us to define the following subset of $N$, which we call the \emph{leading $\cL$-stabilizer} of $\cC$:
\[
  S_{\cC}= \stab_{\cL}(J_\cC).
\]
Given an equation $E = (U = V,\cX, \mu) \in \cC$, we determine maximal prefixes $\wt{U}$ and $ \wt{V}$ of $U$ and $V$, respectively, such that every proper prefix consists only of symbols $\alpha \in \Omega$ with $\mu(\alpha) \in S_\cC$. 
We call the resulting pair $(\wt{U}, \wt{V})$ the \emph{playground} of $E$.
The \emph{size} of this playground is $\wt{n}_E = \abs{\wt{U}\wt{V}}$, and the \emph{players} of $E$ are the variables belonging to 
\[
  \wt{\cX}_E = \set{X \in \cX }{ \cJ(\mu(X)) = J_\cC, \; |L_{\mu,X}| = \infty, \text{ and $X$ occurs twice in $\wt{U}\wt{V}$}}.
\] 
The following establishes that both $\wt{n}_E$ and $\wt{\cX}_E$ are invariants of $\cC$.

\begin{lemma} \label{lem:playground}
  Let $E = (U = V, \cX, \mu) \to E' = (U' = V', \cX', \mu')$ be a transition in~$\cC$ which is labeled by the substitution $\tau$ defined by $X \mapsto \alpha X$.
  Then we have $\wt{\cX}_E = \wt{\cX}_{E'}$ and $\wt{n}_E = \wt{n}_{E'}$.  
  Moreover, if $|UV|_X \geq 1$, then $X \in \wt{\cX}_E$ is a player of $E$ and $\mu(\alpha) \in S_\cC$.
\end{lemma}
\begin{proof}
The claim is trivial in case $|UV|_X = 0$.
Hence, we may assume that $U = Xu = U'$ and $V = \alpha V'$.
Item $(3)$ in \prref{lem:sntrans} yields $\mu'(X) \sim_\cL \mu(\alpha)\mu'(X)$. 
Since $S\in \pv{DLG}$ by hypothesis, the third property in \prref{def:DLRG} implies $\mu(\alpha) \in \stab_{\cL}(\mu'(X))$.
Thus, $\mu(\alp) \in S_{\cC} = \stab_{\cL}(J_\cC)$ as $\cJ(\mu'(X)) = \cJ(\mu(X)) = J_\cC$.

In particular, the symbol $\alpha$ is part of the playground of~$E$, and therefore contributes to the size $\wt{n}_E$ and possibly to the players $\wt{\cX}_E$.
Considering all possible ways in which the variable $X\in \cX$ might occur in $E$ leads us to the following case distinction.
\begin{itemize}
  \item Suppose that $X$ occurs only once in $E$.
    Then we can deduce that $\abs{U'V'} < \abs{UV}$.
  \item Suppose that $X$ occurs again, but outside of the playground of $E$. 
    Then the transition removes $\alpha$ from the playground of $E'$ and thus $\wt{n}_{E'} = \wt{n}_E - 1$ and $\wt{\cX}_{E'} = \wt{\cX}_E \setminus \{\alpha\}$.
  \item Suppose that $X$ occurs again inside the playground of $E$. 
    Then $X$ is a player because, firstly, $\cJ(\mu(X))=\cJ_C$ and, secondly, $\abs{L_{\mu,X}} = \infty$ by $(4)$ in \prref{lem:sntrans}. 
    Moreover, as $\mu(X) = \mu(\alpha)\mu'(X)$ and $\mu(\alpha) \in S_\cC = \stab_{\cL}(\mu'(X))$, we have $\mu'(X) = \mu(\alpha)^{\omega - 1}\mu(X)$. 
    Hence, the sets $L_{\mu, X}$ and $L_{\mu', X}$ are in bijection, so that $\abs{L_{\mu', X}} = \infty$.
  Now let $(X\wt{u}, \alpha\wt{v})$ be the playground of $E$. Since~$X$ is a player and~$\mu(\alpha) \in S_\cC$, we then obtain the playground of~$E'$ as $(X\tau(\wt u), \tau(\wt v))$. In particular, this implies $\wt{n}_{E'} = \wt{n}_E$ and $\wt{\cX}_{E'} = \wt{\cX}_E$.
 \end{itemize}
Finally, since the above case distinction is exhaustive, and since the transition takes place within the strongly connected component~$\cC$, we can rule out the first two cases.
\end{proof}

Since $\wt{\cX}_{E} = \wt{\cX}_{E'}$ for all states $E, E'$ in $\cC$ by \prref{lem:playground}, we can now talk about the players of the strongly connected component~$\cC$, the set of which we denote by $\wt{\cX}_\cC$.
The lemma also states that, within~$\cC$, we may restrict our analysis to the respective playgrounds and treat all non-players as if they were constants.
Be aware however that, even though the set of players~$\wt{\cX}_\cC$ depends only on $\cC$, the playground itself and the arrangement of players within it may vary between states.
For an equation $E = (U = V,\cX,\mu) \in \cC$ with playground $(\wt{U}, \wt{V})$, we split the set of players $\wt{\cX}_E$ into a disjoint union of \emph{balanced} and \emph{unbalanced} players, where a player $X \in \wt{\cX}_E$ is \emph{balanced} if $|\wt{U}|_X = |\wt{V}|_X = 1$, and \emph{unbalanced} if either $|\wt{U}|_X = 2$ or $|\wt{V}|_X = 2$.
As for the playground itself, the notion of balanced and unbalanced players depends on the particular equation $E$ within the strongly connected component $\cC$.

\begin{lemma} \label{lem:cycle}
  Every cycle in the strongly connected component~$\cC$ contains an equation which is nicely balanced with regard to some variable. 
\end{lemma}
\begin{proof}
We begin with a crucial observation:
If one side of an equation $E \in \cC$ begins with a balanced player $X \in \cX$, then~$E$ is nicely balanced with regard to~$X$.
Indeed, by definition of a player we then have $\abs{L_{\mu, X}} = \infty$ and $J_\cC = \cJ(\mu(X))$; in particular, $S_\cC = \stab_{\cL}(\mu(X))$.
Moreover, writing $U = Xu$ and $V = vXv'$, we note that $vX$ is a prefix of $\wt{V}$ and, hence, every symbol~$\alpha \in \Omega$ occurring in $v$ satisfies $\mu(\alpha) \in \stab_{\cL}(\mu(X))$.
This implies $\mu(v) \in \stab_{\cL}(\mu(X))$ by the first item in \prref{def:DLRG}.
Hence, $E$ and $X$ satisfy the requirements of \prref{def:nice}.

Now take some transition~$E \rightarrow E'$ on a cycle within $\cC$.
By item~$(1)$ of \prref{lem:sntrans}, this transition is then labeled by a substitution $\tau$ defined by $X \mapsto \alpha X$ for some $X \in \cX$ and~$\alpha \in \Omega$.
We first note that if~$X$ does not occur in~$E$, then~$E$ is nicely balanced with regard to $X$ since~$\abs{L_{\mu, X}} = \infty$ by item~$(4)$ of the same lemma.
Otherwise, by \prref{lem:playground}, the variable~$X$ is a player and we may assume it to be unbalanced by the observation in the first paragraph.

	Following the substitution~$\tau$, the side not containing~$X$ decreases in length, which can only happen finitely many times.
  Thus, following the cycle, we must at some point encounter a transition~$F \rightarrow F'$ labeled~$Y \mapsto XY$; here~$Y \in \cX \sm \os{X}$ is another variable and moreover, by the same arguments, an unbalanced player.
  This substitution makes~$X$ balanced in~$F'$, so, in order to return to our starting point~$E$, it must become unbalanced once again.
  However, the only way for a balanced player to become unbalanced, is if it occurs at the beginning of one of the sides as a balanced player. 
  This completes the proof.
\end{proof}

Letting now $E$ be any quadratic equation with infinitely many solutions, \prref{prop:matiya} tells us that $\cA(E)$ contains some nontrivial strongly connected component.
Using \prref{lem:cycle}, we can then reduce $E$ to a solvable, nicely balanced equation.
Using \prref{lem:nice}, we finally deduce our main result, which is stated as \prref{thm:main} in the introduction.

\begin{theorem}\label{thm:mainiac}
  Let $(U = V, \mu)$ be a quadratic equation with regular constraints in a semigroup from \pv{DLG}.
  Then $(U = V, \mu)$ has infinitely many solutions if and only if $\exp(U = V, \mu) = \infty$.
\end{theorem}

\subsection{\texorpdfstring{The Brandt Semigroup $B_2$}{The Brandt Semigroup B2}}\label{sec:brandt}

The \emph{(combinatorial) Brandt semigroup} $B_2$ is the semigroup $\os{a,b,ab,ba,0}$ with multiplication given by the equations $\os{aba = a, bab = b, a^2 = b^2 = 0}$.
It does not belong to the class \pv{DLG} because we have $ba\siL a$ but $b^\oo a =0\neq a$, so that property $(3)$ in \prref{def:DLRG} is violated.\footnote{Note that $ba\siL a$ implies that the $\cJ$-class of $a$ contains $ba$. Since $ba$ is an idempotent, the $\cJ$-class of $a$ is regular, but not a subsemigroup as $a^2 = 0 \not\sim_{\cJ} a$.
This means that $B_2$ is not contained in the variety~$\pv{DS}$ consisting of all finite semigroups where regular $\cD$-classes are semigroups.
In fact, it holds that a finite semigroup $S$ belongs to $\pv{DS}$ if and only if $B_2$ does not divide $S \times S$ \cite[Exercise 8.1.6]{alm94}.}
However, we can still show the following.

\begin{proposition}\label{prop:brandt}
  The Brandt semigroup $B_2$ is nice for the family of quadratic word equations with at most two constants.
\end{proposition}

Our proof of \prref{prop:brandt} relies on the following lemma, which itself can be easily verified by routine calculations.
We omit the details of these calculations.

\begin{lemma}\label{lem:brandt-ssg}
  Every proper subsemigroup of $B_2$ belongs to $\pv{DLG}$.
\end{lemma}

\begin{proof}[Proof of \prref{prop:brandt}]
  Let $E = (U = V, \mu)$ be a quadratic equation with $\mu: \Omega^+ \to B_2$ over a set of constants~$\Sig$ with $\abs{\Sig} \leq 2$.
  If $\mu(\Sig)$ generates a proper subsemigroup, then the result follows from \prref{lem:brandt-ssg} together with \prref{thm:mainiac}.
  Hence, we may assume that $\mu(\Sig)$ generates $B_2$.
  Up to symmetry this leaves only one possibility, namely that $\Sig = \os{a, b}$ gets mapped to $B_2$ as $\mu(a) = a$ and $\mu(b) = b$.
  Accordingly, we have
  \[
    \oi\mu(a) = a(ba)^\ast,\quad
    \oi\mu(b) = b(ab)^\ast,\quad
    \oi\mu(ab) = (ab)^{+},\quad
    \oi\mu(ba) = (ba)^{+}.
  \]
  Since arbitrarily long words in these languages clearly have arbitrarily high exponent of periodicity, we may assume that $\mu(X) \notin \os{a,b,ab,ba}$ and, hence, $\mu(X) = 0$ for all $X \in \cX$.

  On the other hand, we have $\oi\mu(0) = \Sig^\ast\os{a^2, b^2}\Sig^\ast$.
  Phrased differently, if $\sig \in \Sol(E)$, then $\sig(X)$ contains a factor $c_X^2$ with $c_X \in \Sig$ for every variable $X \in \cX$.
  We can guess these constants $c_X$ for all variables $X \in \cX$ simultaneously in the form of a mapping $c: \cX \to \Sig$.
  Applying the substitution $X \mapsto \tau_c(X) = X_1 c_X^2 X_2$ with fresh variables $X_1, X_2$ for all $X \in \cX$, we arrive at a quadratic equation $E_c = (\tau_c(U) = \tau_c(V))$ \emph{without} constraints.

  Allowing for singular solutions of the $E_c$, this \nondet procedure is both sound and complete: every solution $\sig'$ of any $E_c$ gives rise to a solution $\sig = \sig'\tau_c$ of $E$, and every solution $\sig$ of $E$ factorizes as $\sig = \sig' \tau_c$ for some solution $\sig'$ of $E_c$ for some suitable $c: \cX \to \Sig$.

  By the pigeonhole principle, if $E$ has infinitely many solutions, then so does some $E_c$.
  Hence, $\exp(E) \geq \exp(E_c) = \infty$ holds by \prref{prop:trivial}.
\end{proof}

For $\Sigma = \{a, b, c\}$ the prove fails to handle mappings like $\mu(a) = a, \mu(b) = b$, and $\mu(c) = ab$.
These images generate the entire semigroup $B_2$, so we are not able to reduce to a proper subsemigroup.
At the same time, $\oi\mu(ab) = \os{ab, c}^{+}$ contains arbitrarily long words with bounded exponent of periodicity, so we cannot hope to eliminate those variables $X \in \cX$ which satisfy $\mu(X) = ab$ as we did in the proof of \prref{prop:brandt}.

\section{\texorpdfstring{The Variety \pv{DLG}}{The Variety DLG}}\label{sec:ddlg}

Our main result, \prref{thm:mainiac}, concerns semigroups from a certain subvariety of~$\pv{DS}$, which itself consists of all finite semigroups whose regular $\cD$-classes are semigroups, and which was first considered in Sch\"utzenberger's article `\emph{Sur le produit de concat{\'e}nation ambigu}' \cite{sch76}.
The subvariety in question, which is also the subject of this section, is the variety $\pv{DLG}$ comprising all finite semigroups whose regular $\cD$-classes are right groups.
Throughout this section, we adopt this as a definition of $\pv{DLG}$, and show in \prref{sec:lst-dlg} that this definition agrees with the earlier \prref{def:DLRG}~--~that is, both define the same class of finite semigroups.

Before we examine this variety in more detail in this section, let us remark that $\pv{DLG}$ (and indeed any variety of the form $\pv{DV}$) has certain useful closure properties:
It is closed under adjunction of a neutral element (hence, the distinction between treating it as a variety of finite semigroups or as one of monoids is marginal) and, adjunction of a zero element.
Due to the latter, we can reduce any system of word equations with constraints in \pv{DLG} to a single such equation as described in \prref{sec:single}.
The variety is also closed under forming nilpotent extensions, so that nothing is to be gained from combining \prref{thm:ell} and \prref{thm:mainiac}.

\subsection{Right Groups}\label{sec:snsdl}

A finite semigroup~$S$ is called a \emph{right group} if it satisfies the identity $x = y^\omega x$. 
If a monoid~$N$ is a right group, then~$N$ is a group.
However, there are infinitely many finite semigroups satisfying $x = y^\oo x$ without being groups. 
To see this, let $G$ be a finite group and $I$ be a nonempty finite set.
Then the set $S = G \times I$ becomes a right group when endowed with the multiplication defined by $(g, i) \cdot (h, j) = (gh, j)$, but $S$ is a group only if $I$ is a singleton.

It is a well-known fact, which however will not be used explicitly here, that every right group is of the above form \cite[Thm.~1.27]{clifford61}.
Instead, we content ourselves with the following.

\begin{lemma}\label{lem:snRss}
A finite semigroup $S$ is a right group if and only if it is right simple, meaning that $S$ has no proper right ideals or, equivalently, that all elements of $S$ are $\cR$-equivalent.
\end{lemma}
\begin{proof}
  If $S$ is a right group, then $y \geq_\cR y^\omega x = x$ for all $x,y \in S$ and, hence, $x \geR y$, too.

  To see the converse, consider $x,y \in S$ and note that $x \sim_\cR y^\omega$ by assumption.
  We may thus write $x = y^\omega z$ for some $z \in S^1$ and, therefore, $x = y^\oo z = y^\oo y^\oo z = y^\oo x$.
\end{proof}

Right groups are regular semigroups, meaning that every element of a right group is regular (since it is contained in a subgroup).
Moreover, the idempotent elements of a right group always form a subsemigroup and, in fact, a right zero semigroup (as witnessed by the identity $y^\oo x^\oo = x^\oo$).
Therefore, right groups are examples of \emph{orthodox} semigroups, which are precisely those regular semigroups in which the idempotents form a subsemigroup.

\subsection{Other Characterizations}\label{sec:Membership}

We now turn to the variety $\pv{DLG}$.
As mentioned in the introduction to this section, it comprises all finite semigroups whose regular $\cD$-classes are right groups.
In particular, $\pv{DLG}$ is contained in the variety $\pv{DO}$ consisting of all finite semigroups whose regular $\cD$-classes are orthodox semigroups. 
The latter is, in turn, contained in the variety $\pv{DS}$ consisting of all finite semigroups whose regular $\cD$-classes are semigroups.
Thus, $\pv{DLG} \sse \pv{DO} \sse \pv{DS}$.

Even though the variety of interest is defined in terms of right groups, its name $\pv{DLG}$ is more aptly explained by the following characterization.\footnote{\prref{prop:Prop749edam} permits reading the second statement as `\emph{in each regular $\cD$-class, each $\cL$-class is a group}', thereby giving rise to a memo for the notation: \textbf{DLG}${}={}$\textbf{D}-class \textbf{L}-class \textbf{G}roup.}

\begin{lemma}\label{lem:memo}
The following two assertions are equivalent for a finite semigroup~$S$. 
\begin{bracketenumerate}
\item Every regular $\cD$-class of $S$ is a right group.
\item Every regular $\cL$-class of $S$ is a group.
 \end{bracketenumerate}
\end{lemma}
\begin{proof}
  \textit{Ad} $(1) \Rightarrow (2)$. 
  Let $\cL(e)$ be a regular $\cL$-class.
  Then $\cD(e)$ is regular and thus a right group by assumption.
  By \prref{lem:snRss}, this implies that all elements of $\cD(e)$ are $\cR$-equivalent.
  But then $\cL(e)$ is a regular $\cH$-class and, in particular, a group.

  \textit{Ad }$(2) \Rightarrow (1)$. 
  Let $\cD(e)$ be a regular $\cD$-class.
  Then $\cL(x)$ is regular for every $x \in \cD(e)$ by \prref{prop:Prop749edam}.
  Hence, by assumption, each such $\cL(x)$ is a group and thus an $\cH$-class.
  This implies $\cD(e) = \cR(e)$.
  To see that $\cD(e)$ is closed under multiplication, note that $xy \sim_\cR x^2 \sim_\cL x$ for $x, y \in \cD(e)$.
  The first equivalence follows since $x \sim_\cR y$ and the second, since $\cL(x)$ is a group.
  In particular, we have $x \sim_\cR y^\omega$. 
  Therefore we may write $x=y^\omega z$ for some $z \in S^1$, which implies $y^\omega x=x$.
  Hence, $\cD(e)$ is a right group.
\end{proof}

The variety \pv{DLG}, along with its left-right dual \pv{DRG}, has been studied extensively by Almeida and Weil~\cite{AlmeidaWeil1997IJAC}.
Various other results on \pv{DLG} and its generalizations can be found in C\'{e}lia Borlido's thesis~\cite{BorlidoPHD2016} and the articles \cite{AlmeidaBorlidoIJAC2017,BorlidoTCS2017,BorlidoComAlg2018}. 
The following characterization of \pv{DLG} and its proof were kindly provided by Jorge Almeida (personal communication).

\begin{lemma}\label{lem:dlgth}
For every finite semigroup $S$, the following assertions are equivalent. 
\begin{bracketenumerate}
 \item The semigroup $S$ satisfies the identity $(xy)^\oo= y^\oo (xy)^\oo$.
  \item The semigroup $S$ satisfies the identity $(xyz)^\omega = y^\omega (xyz)^\omega$.
  \item The semigroup $S$ satisfies the identity $(xy)^{\oo} = (yx)^{\oo} (xy)^{\oo}$.
  \item Every regular $\cD$-class of $S$ is a right group.
\end{bracketenumerate}
\end{lemma}
\begin{proof}
  To see the implication $(1) \Rightarrow (2)$, let us first note that $(1)$ implies the identity
  \[
  (xy)^{\oo} = x (yx)^{\oo} y(xy)^{\oo-1}
   = x x^{\oo}(yx)^{\oo} y(xy)^{\oo-1}
	 = x^{\oo} x(yx)^{\oo} y(xy)^{\oo-1}
	 = x^{\oo} (xy)^{\oo}.
  \]
  Using this, we obtain $(xyz)^{\oo} = (yz)^{\oo} (xyz)^{\oo} = y^{\oo} (yz)^{\oo} (xyz)^{\oo} =y^{\oo} (xyz)^{\oo}$.

  Assuming $(2)$, the semigroup $S$ satisfies $(xy)^{\oo} = (xyxy)^{\oo} = (yx)^{\oo} (xyxy)^{\oo} = (yx)^{\oo} (xy)^{\oo}$, which shows that $(2) \Rightarrow (3)$.
  To see $(3) \Rightarrow (1)$, first note that $(xy)^{\oo} = (yx)^{\oo} (xy)^{\oo}$ implies the defining identity $(xy)^{\oo} = ((xy)^{\oo}(yx)^{\oo}(xy)^{\oo})^{\oo}$ for the variety \pv{DS} \cite[Sec.~8.1]{alm94}.
  Now, for all $x,y \in S$, the idempotents $e = (y^{\oo} (xy)^{\oo})^{\oo}$ and $f= (xy)^{\oo}$ satisfy $e\sim_{\cJ} f$ by \cite[Lemma~8.1.4]{alm94}.
  Indeed, we have $f \leJ y^n$ for all $n \geq 1$ since this implies $f \siJ fy^n \leJ y^{n+1}$ by the cited lemma and the assertion clearly holds for $n = 1$.
  So $f \leJ y^{\oo}$ and, therefore, $f \siJ y^{\oo} f \siJ (y^\oo f)^\oo = e$.

  Moreover, since $e \leq_{\cL} f$, this shows $e \sim_{\cL} f$ so that $e = ef$ and $f = fe$.
  We then deduce that $f = fe = (ef)(fe) = ef = e$ using $(xy)^{\oo} = (yx)^{\oo} (xy)^{\oo}$ for the second equality.
  Hence, it follows that $(xy)^{\oo} = (y^{\oo} (xy)^{\oo})^{\oo} = y^{\oo} (y^{\oo} (xy)^{\oo})^{\oo} = y^{\oo} (xy)^{\oo}$ as required.

  Still assuming $(3)$, we know from the above that $S \in \pv{DS}$.
  Hence, if $x,y \in S$ are elements of the same regular $\cD$-class, then $xy \siJ yx \siJ x$.
  Thus, $x = uyxv = (uy)^{\oo} x v^{\oo}$ for some $u,v \in S^1$.
  Applying the identity from $(1)$ yields $x = (uy)^{\oo} x v^{\oo} = y^{\oo} (uy)^{\oo} x v^{\oo} = y^{\oo} x$, which is a defining identity for right groups; hence, $(3) \Rightarrow (4)$.
  Finally, to see that $(4) \Rightarrow (3)$, note that $(xy)^{\oo}$ and $(yx)^{\oo}$ are always contained in the same (necessarily regular) $\cD$-class. 
  Due to the identity $x = y^{\oo}x$ for right groups, this implies the equation $(xy)^{\oo} = (yx)^{\oo} (xy)^{\oo}$.
\end{proof}

\subsection{\texorpdfstring{Properties of $\cL$-Stabilizers}{Properties of L-Stabilizers}}\label{sec:lst-dlg}

We now turn to the $\cL$-stabilizers discussed in \prref{sec:lst}, which are particularly well-behaved for finite semigroups belonging to $\pv{DLG}$.
Indeed, the following results show that the definition of the class \pv{DLG} used throughout this section is consistent with \prref{def:DLRG}.

\begin{lemma}\label{lem:lst2}
  Let $S \in \pv{DLG}$.
  Then the following hold for all $x,y \in S$ and $u,v \in S^1$.
  \begin{bracketenumerate}
    \item If $u^{\oo} x = x$ and $v^{\oo} x = x$, then $(uv)^{\oo} x = x$; that is, $\stab_{\cL}(x)$ is a submonoid of $S^1$.
    \item If $x \siJ y$, then $u^{\oo} x = x$ if and only if $u^{\oo}y = y$; that is, $x \siJ y \Rightarrow \stab_{\cL}(x) = \stab_{\cL}(y)$.
  \end{bracketenumerate}
  Conversely, if the second item is satisfied for some finite semigroup $S$, then $S \in \pv{DLG}$.
\end{lemma}

\begin{proof}
  \textit{Ad $(1)$.} Suppose that $u^{\oo}x = x$ and $v^{\oo} x = x$.
  Then $x = u^{\oo}x = u^{\oo}v^{\oo} x = (u^{\oo}v^{\oo})^{\oo}x$ and, since $uv \geJ u^\oo v^\oo$,  item $(2)$ of \prref{lem:dlgth} yields $(u^{\oo}v^{\oo})^{\oo}x = (uv)^{\oo}(u^{\oo}v^{\oo})^{\oo}x = (uv)^{\oo} x$.

  \textit{Ad $(2)$.} Suppose that $x \siJ y$, that is, $y = rxs$ and $x = r'ys'$ for some $r,r',s,s' \in S^1$.
  If, moreover, $u^\oo x = x$, then $y = ru^{\oo} r' y s' s = (ru^{\oo}r')^{\oo} y (s's)^{\oo} = u^{\oo} (ru^{\oo}r')^{\oo} y (s's)^{\oo} = u^{\oo} y$ where the second to last equality is due to item $(2)$ of \prref{lem:dlgth}.
  Conversely, if $S$ satisfies $(2)$, then we have $yx \in \stab_{\cL}((yx)^{\oo}) = \stab_{\cL}((xy)^{\oo})$ for all $x,y \in S$ since $(yx)^{\oo} \siJ (xy)^{\oo}$.
  Hence, the semigroup $S$ satisfies the identity $(xy)^{\oo} = (yx)^{\oo} (xy)^{\oo}$ so that $S \in \pv{DLG}$ by \prref{lem:dlgth}.
\end{proof}

\begin{lemma}\label{lem:lst1}
  Let $S \in \pv{DLG}$.
  Then $u,x \in S$ satisfy $u^\oo x = x$ if and only if $ux \siL x$.
  Conversely, if this equivalence holds for all elements $u, x$ of a finite semigroup $S$, then $S \in \pv{DLG}$.
\end{lemma}
\begin{proof}
  If $u^{\oo} x = x$, then $x = u^{\oo} x \leL ux \leL x$ so that $ux \siL x$.
  Moreover, $ux \siL x$ holds if and only if $x = vux$ and, hence, $x = (vu)^{\oo}x$ for some $v \in S^1$.
  Hence, we also have $(vu)^{\oo} = u^{\oo}(vu)^{\oo}$ by the first identity from \prref{lem:dlgth} and, therefore, $x = (vu)^{\oo}x = u^{\oo}(vu)^{\oo}x = u^{\oo}x$.

  Now suppose that $ux \siL x \Rightarrow u^\oo x = x$ holds in a finite semigroup $S$.
  Then, for all $x, y \in S$, we have $y^\oo (xy)^\oo = (xy)^\oo$ since $y (xy)^\oo \siL (xy)^\oo$.
  Hence, $S \in \pv{DLG}$ by \prref{lem:dlgth}.
\end{proof}

\subsection{Examples}\label{sec:ATexas}

The following non-exhaustive list shows that the variety \pv{DLG} is rather rich in the sense that it contains various interesting families of semigroups and can recognize interesting languages. 
\begin{itemize}
  \item For aperiodic semigroups, the variety \pv{DLG} contains the variety \pv{Sl} of finite semilattices, which allows for the implementation of \emph{alphabetical constraints}, as well as the larger variety \pv{J} of finite $\cJ$-trivial semigroups.
    The latter allows for \emph{piecewise testable constraints} in the sense of Imre Simon \cite{sim75} (see also \cite[Thm.~4]{pin86} or \cite[Thm.~4]{DiekertGastinKufleitner08}).
  \item More generally, \pv{DLG} contains the variety $\pv{RRB} \supseteq \pv{Sl}$ of finite right regular bands, as well as the larger variety $\pv{L} \supseteq \pv{J}$ of finite $\cL$-trivial semigroups, which coincides with the class of all finite semigroups whose regular $\cD$-classes are right zero semigroups. 
    On the other hand, the (nontrivial) left zero semigroups are \emph{not} contained in \pv{DLG}.
  \item As a subvariety of \pv{J}, the variety \pv{N} of all finite nilpotent semigroups is contained in \pv{DLG}, and so is the variety $\pv{D} \sse \pv{L}$ consisting of all finite semigroups in which all idempotents are right zeros.
    Using constraints in \pv{N} it is possible to enforce (or prohibit) the length of a solution to exceed a specified threshold, or to require that the solution is a specific word;
    using constraints in \pv{D} one can require a specific word to appear as a suffix.\footnote{For equations, these types of constraints can also be realized via substitutions and case distinctions.}
  \item The variety \pv{G} of finite groups is also clearly contained in \pv{DLG}, which allows for \emph{group constraints}, such as modular counting.
    Note also that this type of constraints may be quite liberally combined with those discussed in the previous items; see \cite[Prop.~2.2.4]{AlmeidaWeil1997IJAC}.
\end{itemize}

Another interesting family of finite semigroups contained in \pv{DLG} is the family of \emph{duo} semigroups comprising semigroups $S$ in which $x S^1 = S^1 x$ holds for all $x \in S$.
This class of semigroups, which is not a variety but contains all groups and commutative semigroups, was defined by Pond\v{e}l\'{i}\v{c}ek~\cite{PondelicekCzMathJ1975} as a semigroup-theoretic generalization of Feller’s notion of duo rings~\cite{FellerTransAMS1958}.
Note that in a duo semigroup $S$ every product of the form $(x_1y)(x_2y) \cdots (x_n y)$ with $x_1, \dotsc, x_n, y \in S$ can be written as $y^n z$ for some $z \in S^1$.
Hence, if $S$ is finite, then we obtain $(xy)^{\oo} = y^{\oo} z = y^{\oo} y^{\oo} z = y^{\oo} (xy)^{\oo}$ for some $z \in S^1$. 
Thus, $S \in \pv{DLG}$ by \prref{lem:dlgth}.

Apart from the nontrivial left zero semigroups mentioned above, we have seen in \prref{sec:brandt} that the Brandt semigroup $B_2$ does \emph{not} belong to $\pv{DLG}$ and, in fact, not even to $\pv{DS}$.

\section{Conclusion and Open Problems}\label{sec:cop}

The guiding question of this work is whether a word equation with infinitely many solutions must have an infinite exponent of periodicity.
We do not yet know the answer. 
Nevertheless, regardless of the ultimate resolution of this question, the results obtained here~--~focused on equations with regular constraints~--~are of independent interest.

Our approach has two complementary aspects, like Yin and Yang or two sides of the same medal: on one side, a class $\cE$ of word equations; on the other, a class $\pv{V}$ of semigroups.
When $\pv{V}$ is the variety generated by the trivial semigroup (that is, for equations without constraints), it was already known that a word equation with infinitely many solutions has an infinite exponent of periodicity in two cases: two-variable and quadratic equations.

The quadratic case extends naturally to constraints in finite groups. More generally, if $U = V$ is quadratic and the constraint $\mu$ maps variables to units of a finite monoid $M$, the machinery from Section~\ref{sec:origSN} applies because finite monoids are Dedekind finite, meaning that their non-units form an ideal.
In this setting, the playground of a word $W \in \Omega^+$ is the maximal prefix $\widetilde{W}$ that is mapped to a unit $\mu(\widetilde{W})$. 
Similar reasoning works for alphabetic constraints, where playgrounds are defined in terms of sets of constants. 
This approach extends to duo semigroups and to the variety $\mathbf{DLG}$ (and its left-right dual $\mathbf{DRG}$), but seems to reach its limit there.
Indeed, $\pv{DLG}$ is up to duality the largest variety of finite semigroups known (to us) whose members are nice for all quadratic word equations.
Whether an even larger variety shares this property remains open.
Beyond this, a few special cases are known: for example, we have seen that the Brandt semigroup $B_2$ is nice for quadratic equations with at most two constants. 

Still, even without constraints, the central question remains far from settled, and there is ample room for partial progress. 
Below we outline several open problems. 
In what follows, it is natural to first restrict attention to $\cE$ consisting of two-variable or quadratic equations.

\begin{problem}\label{prob:1}
  Does the family of nice semigroups for $\cE$ form a variety of finite semigroups? 
\end{problem}

\begin{problem}\label{prob:2}
  Given a finite semigroup $S$, is it decidable whether $S$ is nice for $\cE$?
\end{problem}

\begin{problem}\label{prob:3}
  Is every finite semigroup in $\pv{DO}$ or $\pv{DS}$ nice for $\cE$?
\end{problem}

The equation $XabY = YbaX$ admits solutions where $\sigma(X)$ and $\sigma(Y)$ are arbitrarily long factors of Fibonacci words \cite[Sec.~6]{WeinbaumPacific2004ABCD}, hence with bounded exponent of periodicity, making it a candidate for a potential counterexample. 
Regular constraints alone might not suffice, so one could also consider length constraints~--~for instance, requiring $p \cdot \abs{\sig(Y)} \leq \abs{\sig(X)} \leq q \cdot \abs{\sig(Y)}$ for rationals $p < q$ close to the golden ratio $\phi=1.618\dots$.
While the decidability of word equations with length constraints is a long-standing open problem, Day and Konefa~\cite{RP2025DayKonefa} recently proved it decidable for two-variable quadratic equations, including $XabY = YbaX$.

\begin{problem}\label{prob:4}
  Consider an equation $XAY = YBX$ over $\OO = \os{a,b} \cup \os{X,Y}$ with $A,B \in \OO^+$, for example $XabY = YbaX$.
  Do there exist a $\kappa\in \N$ and a combination of regular constraints and length constraints such that, under these constraints, the equation has infinitely many solutions, and every solution has exponent of periodicity less than $\kappa$?
\end{problem}

The problems above about word equation in free semigroups also arise naturally in the setting of free groups.
Here, recognizable constraints correspond to homomorphisms to finite groups.
A much larger and much more interesting class of constraints for a free group $F$ is the family of  rational subsets $\Rat(F)$.
It is known that $\Rat(F)$ forms an effective Boolean algebra \cite{ben69}, and that the full solution set of an equation with a rational constraint can be effectively described by an EDT0L language, using results from \cite{dgh05IC}, \cite{DiekertJP16}, and \cite{CiobanuDiekertElder2016ijac}.
However, even for quadratic equations, the following question remains open.

\begin{problem}\label{prob:5}
  Are all rational constraints nice for quadratic word equations in free groups?
\end{problem}

\bibliographystyle{abbrv}
\bibliography{qwe}

\end{document}